\documentclass{article} 

\title{Single-Winner Voting with Alliances: \\
       Avoiding the Spoiler Effect}

\author{Grzegorz Pierczyński$^{1,2}$, Stanisław Szufa$^{1,3}$\\
    $^1$ AGH University, Poland\\
    $^2$ University of Warsaw, Poland\\
    $^3$ Université Paris Dauphine-PSL, France
}
\date{}

\newcommand{\BibTeX}{\rm B\kern-.05em{\sc i\kern-.025em b}\kern-.08em\TeX}

\usepackage[utf8]{inputenc}
\usepackage{amsmath}
\usepackage{natbib}
\usepackage{url}
\usepackage{amsthm}
\usepackage{amssymb}
\usepackage{graphicx}
\usepackage{dsfont}
\usepackage{tikz}
\usepackage{tikz}
\usetikzlibrary{patterns,decorations.pathreplacing, arrows.meta}

\usepackage{xspace}
\usepackage{placeins}

\usepackage{xcolor}
\usepackage{nicefrac}
\usepackage{comment}
\usepackage{paralist}
\usepackage{cleveref}
\usepackage{mathtools}
\usepackage[colorinlistoftodos,textsize=tiny]{todonotes}

\newtheorem{theorem}{Theorem}
\newtheorem{proposition}{Proposition}
\newtheorem{observation}{Observation}

\newtheorem{definition}{Definition}
\newtheorem{example}{Example}

\newtheorem{intuition}{Intuition}
\usepackage{booktabs}

\newcommand{\pref}{{{\mathrm{Pref}}}}

\newcommand{\alliances}{{{\mathcal A}}}

\usepackage{pifont}
\definecolor{myred}{RGB}{201, 22, 22}
\definecolor{mygreen}{RGB}{38, 150, 68}
\newcommand{\yes}{\textcolor{mygreen}{\ding{51}}}
\newcommand{\no}{\textcolor{myred}{\ding{55}}}

\newcommand{\Keyally}{Solitary winner\xspace}
\newcommand{\keyally}{solitary winner\xspace}
\newcommand{\keyxally}{solitary-winner\xspace}
\newcommand{\Keyxally}{Solitary-Winner\xspace}
\newcommand{\keyallyshort}{SW\xspace}
\newcommand{\keyallyplural}{solitary winners\xspace}

\newcommand{\Independentwinner}{Independent winner\xspace}
\newcommand{\independentwinner}{independent winner\xspace}
\newcommand{\independentxwinner}{independent-winner\xspace}
\newcommand{\Independentxwinner}{Independent-Winner\xspace}
\newcommand{\independentwinnershort}{IW\xspace}
\newcommand{\independentwinnerplural}{independent winners\xspace}

\newcommand{\Allynoharm}{Ally-no-harm\xspace}
\newcommand{\allynoharm}{ally-no-harm\xspace}
\newcommand{\Splitting}{Resistance to alliance-splitting\xspace}
\newcommand{\splitting}{resistance to alliance-splitting\xspace}
\newcommand{\splittingadj}{resistant to alliance-splitting\xspace}
\newcommand{\Clonedallynohelp}{Independence of similar allies\xspace}
\newcommand{\clonedallynohelp}{independence of similar allies\xspace}
\newcommand{\clonedallynohelpadj}{independent of similar allies\xspace}

\newcommand{\fixedwidth}[1]{\makebox[1.2em][c]{$#1$}}

\usepackage{multirow}

\usepackage{tabularx}

\begin{document}



\maketitle 

\begin{abstract}
We study the setting of single-winner elections with ordinal preferences where candidates might be members of \emph{alliances} (which may correspond to e.g., political parties, factions, or coalitions). However, we do not assume that candidates from the same alliance are necessarily adjacent in voters' rankings. In such case, every classical voting rule is vulnerable to the spoiler effect, i.e., the presence of a candidate may harm his or her alliance. We therefore introduce a new idea of \emph{alliance-aware} voting rules which extend the classical ones. We show that our approach is superior both to using classical cloneproof voting rules and to running primaries within alliances before the election.

We introduce several alliance-aware voting rules and show that they satisfy the most desirable standard properties of their classical counterparts as well as newly introduced axioms for the model with alliances which, e.g., exclude the possibility of the spoiler effect. Our rules have natural definitions and are simple enough to explain to be used in practice.
\end{abstract}

\section{Introduction}\label{sec:introduction}

Imagine that there is a presidential election in the Republic of Social Choice. There are only two political parties, Party~A and Party~B. Party~B is represented by only one candidate, Bob, while Party~A has two leaders willing to candidate, Alice and Adam. Alice presents more center-wing views, and is therefore considered to have greater chances to win against Bob. On the other hand, Adam is much more popular within the electorate of Party~A. In the end, both candidates decide to run in the election. The voters' preferences are as follows:
\begin{alignat*}{3}
    46\%\text{ of voters:}\quad & \text{ Adam } &&\succ \text{ Alice } &&\succ \text{ Bob } \\
    5\%\text{ of voters:}\quad & \text{ Alice } &&\succ \text{ Bob } &&\succ \text{ Adam } \\
    43\% \text{ of voters:}\quad & \text{ Bob } &&\succ \text{ Alice } &&\succ \text{ Adam } \\
    6\%\text{ of voters:}\quad & \text{ Bob } &&\succ \text{ Adam } &&\succ \text{ Alice }
\end{alignat*}

Note that, apart from the clear supporters of Party~A and Party~B, there are also~${5\%}$ of centrist voters, who, in general, like Party~A more than Party~B, yet they consider Adam to be too extreme, and hence prefer Bob to Adam, even though Bob is from Party~B. On the other hand, there are~${6\%}$ of Bob's supporters who were impressed by Adam's performance in the pre-election debate and consider him to be a better candidate than Alice.

Since the Republic of Social Choice uses Plurality in presidential elections, Bob is declared the winner. The supporters of Party~A are very disappointed, especially since they know that they ``could have won'' if only Adam did not run in the election. Their main question is then how to avoid such situations in the future.

Some of them say that the solution is clear---there should have been only one running candidate from Party~A instead of several ones. However, if Party~A held a primary election, Adam would have won (no matter whether the voters taking part in the primaries represented the views of only the supporters of Party~A or of the whole population). On the other hand, since Party~A presents itself as a transparent and pro-democratic party, the backstage nomination of Alice was completely out of the question.

An alternative opinion is that the problem stemmed from using Plurality to determine the winner. ``In our elections, we should use a more sophisticated rule that is independent of clones'', it is argued, ``for example, STV or the Schulze's method''. However, the proposed rules---as well as every other well-established ordinal voting rule the supporters of Party~A are aware of---would still elect Bob under the above voters' preferences. 

Studying the existing literature on social choice, the supporters of Party~A could not find a fully satisfactory solution to their problem.

\subsection{Our Contribution}
The situation presented above is sometimes called the \emph{spoiler effect}~\citep{borgers2010mathematics}---Adam not only did not win, but also prevented a member of his party from winning. We propose a new idea that solves this problem, \emph{alliance-aware voting rules}, that take as input not only voters' preferences, but also the information about the non-overlapping \emph{alliances} (parties) of the candidates.

In \Cref{sec:axioms}, we propose several basic axioms that are desirable for alliance-aware voting rules. Our basic axioms, discussed in \Cref{sec:basic-axioms}, prevent the spoiler effect among allies and ensure that it is not possible to affect the election result by certain sorts of manipulation (such as duplicating allies, or splitting a losing alliance into smaller ones). We argue that the rules based on the mechanics of Plurality and Maximin are best-suited for satisfying these axioms. Next, in \Cref{sec:individual-axioms} we present additional axioms that concern electing the best candidate within the winning alliance and argue that one could have two reasonable, yet contradictory, intuitions here. Then in \Cref{sec:rules} we present four alliance-aware rules and show that they have good axiomatic properties (i.e., they satisfy our newly introduced notions and preserve the most critical properties of the classical variants of Plurality or Maximin they are based on). 

In \Cref{sec:experiments} we analyze our setting experimentally. We show that the standard variants of our alliance-aware rules often admit the spoiler effect among allies, even if we run primaries within each alliance before the voting. Further, we show that the additional axiomatic guarantees provided by our rules, fully eliminating this problem, do not require sacrificing much of social welfare.

Notably, all the four alliance-aware rules are polynomial-time computable and have simple, intuitive definitions. This fact makes them reasonable proposals for real-life elections. Additionally, in \Cref{sec:conclusion} we show that the extensions of Plurality can be implemented with simpler ballots than the full ordinal ones and discuss how our results can be applied to the model where the alliances are not fully disjoint.

\subsection{Related Work}\label{sec:related-work}


Our work is closely related to two widely studied topics in the literature on social choice---the spoiler effect and strategic candidacy. 

The most important work regarding the spoiler effect is the famous work of \citet{arrow1950difficulty}, where the Arrow's Impossibility Theorem is proved. This theorem states that every reasonable rule (satisfying the very weak conditions of unanimity and non-dictatorship) violates the \emph{Independence of Irrelevant Alternatives} (IIA) axiom. It means that for every reasonable rule the removal of a non-winning candidate~${c}$ may result in changing the result of the election---which means that the rule is vulnerable to the spoiler effect. This work has had tremendous impact on social choice, which manifested itself in the number of works that discuss, criticize, or propose weaker variants of this axiom (see, for example, the works of \citet{ray1973independence,Osborne76,bordes1991independence,Maskin20} and~\citet{Sen70}).

Among such works, the one by \citet{tideman1987independence} is the most relevant to our paper. Here the axiom of \emph{independence of clones} was introduced. The idea is to focus only on avoiding the spoiler effect within groups of \emph{clones} that intuitively correspond to candidates that are, in a certain formal sense, similar. We find this weakening natural and intuitive. However, in that work the similarity is deduced from the votes---the clones are required to be adjacent in all the voters' rankings. Note that in our example, since a small fraction of voters does not rank Alice and Adam next to one another, this axiom does not provide any guarantees regarding vote-splitting. On the contrary, we assume that the information about the candidates' similarity is given upfront.

So far only few papers studied the immunity of certain voting rules to the spoiler effect; a notable exception is the analysis of STV in the work of~\citet{mccune2023ranked}. From the computational complexity perspective, the spoiler effect has been studied within so called \emph{electoral control problems}~\citep{meir2008complexity,liu2009parameterized,chevaleyre2010possible,neveling2020complexity,erdelyi2021towards}.

The next group of related works shares with us the assumption that the information about candidates' views is given, yet typically in the form of rankings over other candidates. Then the problem of the spoiler effect is studied from the strategic perspective (see the pioneering paper of \citet{dutta2001strategic} who prove an analogous result to the Arrow's one in this setting, and further follow-ups e.g., by~\citet{ehlers2003candidate,polukarov2015convergence,lang2013new} and~\citet{rodriguez2006candidate,rodriguez2006correspondence}).

Some of the most closely connected works to ours are the ones by~\citet{faliszewski2016hard} and by~\citet{harrenstein2021hotelling}. The first one studies the computational complexity of problems connected with selecting the optimal representative of a party in single-winner elections, while the second one characterizes Nash equilibria in a generalized Hotelling-Downs model.

\section{Preliminaries}
An \emph{election} is a tuple~${(C, V, \alliances)}$ where~${C=\{c_1,\dots,c_m\}}$ is the set of~${m}$ candidates,~${V=\{v_1,\dots,v_n\}}$ is the set of~${n}$ voters, and~${{\alliances\subseteq 2^C}}$ denotes the set of \emph{alliances} between candidates. Each voter~${v_i\in V}$ is associated with a strict linear order (ranking)~${\succ_i}$ over the candidates. For each pair of candidates~${a, b\in C}$, we let~$\pref(a, b)$\linebreak$=|\{v_i\in V\colon a \succ_i b\}|$ and say that~${a}$ \emph{wins head-to-head against~${b}$} if~${\pref(a, b) > 
\nicefrac{n}{2}}$. We assume that~${|\alliances| \geq 2}$ and that every candidate belongs to exactly one alliance (possibly a single-element one), i.e., we have that~${\bigcup_{A\in \alliances} A = C}$ and for each pair of alliances~${A_1, A_2\in \alliances}$ we have that~${A_1 \cap A_2 = \emptyset}$. For each candidate~${c\in C}$, by~${A(c)}$ we denote the alliance containing~${c}$. We say that candidates~${a,b \in C}$ are \emph{allies} if they belong to the same alliance, i.e.,~${A(a) = A(b)}$; otherwise, we say that~${a}$ and~${b}$ are \emph{opponents}. 

We consider two special cases of our model:
\begin{itemize}
    \item \emph{two-alliance} elections, where~${|\alliances| = 2}$,
    \item \emph{no-ally} elections, where~${|\alliances| = m}$ (i.e., for every~${A\in \alliances}$ we have that~${|A|=1}$).
\end{itemize}
Additionally, we say that a voter~${v_i\in V}$ is \emph{alliance-consistent} if for each two allies~${a_1, a_2}$ there exists no candidate~${b \notin A(a_1)}$ such that~${a_1 \succ_i b \succ_i a_2}$.

A \emph{voting rule} (or, a \emph{rule}, for short) is a function that, for a given election, returns a winning candidate~${c}$.\footnote{Typically, definitions of voting rules allow returning a set of tied winners. We assume that ties are broken lexicographically. The choice of a (deterministic) tie-breaking rule does not affect our results.} We say that a voting rule~${f}$ is \emph{standard} if it does not use the information about alliances; i.e., for every two elections~${E_1=(C, V, \alliances)}$ and~${E_2 = (C, V, \alliances')}$ we have that~${f(E_1) = f(E_2)}$. If a voting rule is not standard, we say that it is \emph{alliance-aware}.

Given an election~${E}$, we sometimes say that an alliance~${{A\in \alliances}}$ wins under~${f}$ if a member of~${A}$ wins under~${f}$, i.e.,~${f(E)\in A}$. Analogously, an alliance~${A\in \alliances}$ loses under~${f}$ if~${f(E)
\notin A}$.

\subsection{Standard Axioms of Voting}

We consider the following well-known voting axioms, proposed in the literature for standard voting rules:

\begin{description}
\item[Majority consistency.]
A \emph{majority winner} is a candidate who is ranked on top by a majority (more than~${
\nicefrac{n}{2}}$) of the voters. We say that a voting rule is \emph{majority-consistent} if it elects the majority winner whenever one exists.
\item[Condorcet consistency.]
A \emph{Condorcet winner} is a candidate who wins head-to-head with all the other candidates. We say that a voting rule is \emph{Condorcet-consistent} if it elects the Condorcet winner whenever one exists.

\end{description}
The next axioms are based on the idea that ``modifying'' a considered election in some way should neither ``harm'' nor ``help'' certain candidates. For a fixed voting rule~${{f}}$, we say that a candidate~${{c\in C}}$ is \emph{harmed} by modifying election~${E}$ into~${E'}$ if~${f(E)=c}$ and~${f(E')\neq c}$. Similarly, we say that~${c}$ is \emph{helped} if~${f(E)\neq c}$ and~${f(E')=c}$.
\begin{description}
\item[Monotonicity.]
We say that a voting rule~${f}$ is \emph{monotone} if no candidate can be harmed by improving his or her position in rankings of some voters (without changing the relative order of the other candidates).

\item[Independence of clones (cloneproofness).] Consider an election~${E}$. We say that a group of candidates~${T\subsetneq C}$ are \emph{clones} if there are no candidates~${a, b\in T}$, and~${c \notin T}$ such that for some voter~${v_i\in V}$ it holds~${a \succ_i c \succ_i b}$. Intuitively, candidates~${T}$ are adjacent in all voters' rankings. Consider now election~${E'}$ obtained by removing a candidate from~${T}$ from election~${E}$. A voting rule~${f}$ is \emph{independent of clones} (or, \emph{cloneproof}) if the following two conditions are satisfied:
\begin{enumerate}[(1)]
     \item if~${f(E)\in T}$ then~${f(E')\in T}$,
     \item if~${f(E) \notin T}$ then~${f(E) = f(E')}$. 
\end{enumerate}
\end{description}

Intuitively, clones are candidates that are considered by the voters to be similar to one another. The axiom states that removing candidates from a set of clones~${T}$ should neither harm nor help candidates outside of~${T}$. 

\subsection{Standard Voting Rules}

In our work we consider the following well-established standard voting rules:

\begin{description}
    \item[Positional scoring rules.] Positional scoring rules elect the candidate maximizing the score, which can be characterized by a vector~${\vec{s} = (\ell_1, \ell_2, \ldots, \ell_m)}$ such that~$\ell_1 \geq \ell_2 \geq \ldots \geq \ell_m$. A candidate~${c}$ gets~${\ell_i}$ points from each voter~${v}$ who ranks~${c}$ at the~${i}$th position. The best-known positional scoring rule is Plurality, characterized by the vector~${\vec{s} = (1, 0, \ldots, 0)}$. We further call the number of voters who rank a candidate~$c$ on top the \emph{plurality score} of~$c$. 


\item[Copeland.] When the number of voters~${n}$ is odd,\footnote{There are several variants of Copeland dealing differently with tied head-to-head comparisons (that is, head-to-head comparisons between candidates~${a}$ and~${b}$ where~${\pref(a, b) = 
\nicefrac{n}{2}}$); the choice of a specific variant for even~${n}$ does not affect our results.} Copeland elects the candidate~${c\in C}$ who wins head-to-head against the greatest number of candidates. 
\item[Maximin.] Maximin elects the candidate~${c\in C}$ maximizing the \emph{maximin score}, defined as~${\min_{c'\in C} \pref(c, c')}$.
\item[STV.] Single Transferable Vote proceeds recursively: In each round it eliminates the candidate with the smallest plurality score until there remains only one candidate.
\end{description}

Additionally, in \Cref{app:schulze} we present the definition and analysis of the Schulze's rule \citep{schulze2011new}. 

Our considered alliance-aware voting rules will \emph{extend} the above standard ones. Formally, an alliance-aware voting rule~${f}$ extends a standard voting rule~${g}$, if for each no-ally election~${E}$ it holds that~${f(E) = g(E)}$.

\section{Axioms of Voting with Alliances}\label{sec:axioms}

We can clearly see that alliance-aware voting rules use more information than standard ones, which might potentially result in electing better candidates. In this section we present axioms that we consider natural for this setting.

\subsection{Basic Axioms}\label{sec:basic-axioms}

Our overall intuition is that the additional information about alliances should be used to avoid vote-splitting between allies and \emph{only} for that purpose. We formalize it via two following axioms:

\begin{definition}[\Allynoharm]\label{def:ally-no-harm}
We say that a voting rule~${f}$ satisfies \emph{\allynoharm} criterion if for every elections~${E, E'}$ such that~${E'}$ is obtained from~${E}$ by removing a candidate~${c}$ we have that
\begin{align*}
    f(E) \notin A(c) \implies f(E') \notin A(c).
\end{align*}
\end{definition}

\begin{definition}[\Splitting]
    We say that a voting rule~${f}$ is \emph{\splittingadj} if for every elections~${E, E'}$ such that~${E'}$ is obtained from~${E}$ by splitting an alliance~${A}$ into two alliances~${A_1, A_2}$ (so that~${A_1\cup A_2 = A}$ and~${A_1\cap A_2 = \emptyset}$) we have that
    \begin{equation*}
        f(E) \notin A \implies f(E) = f(E').
    \end{equation*} 
\end{definition}

Intuitively, the first axiom ensures that a candidate is never a spoiler for members of his or her alliance. The second axiom ensures that the information about alliances is not ``abused'' by an alliance-aware voting rule. The only consequence of splitting an alliance~${A}$ into two alliances~${A_1, A_2}$ should be increasing the level of vote-splitting between members of~${A_1}$ and~${A_2}$, which might result in decreasing their winning chances. Therefore, if no candidate from~${A}$ wins anyway, splitting~${A}$ should not affect the outcome of the rule at all. As a result, \splitting intuitively rules out alliance-aware voting rules that are too artificial (whose definitions depend on, e.g, the number of alliances or their quantities) or that are too easily manipulable.

In order to fully avoid the spoiler effect, we also propose an analogue of the monotonicity axiom for the setting with alliances: improving the level of support of a candidate should never harm his or her alliance.

\begin{definition}[Alliance monotonicity]
     We say that a voting rule~${f}$ is \emph{alliance monotone} if for every elections~${E, E'}$ such that~${E'}$ is obtained from~${E}$ by improving the position of some candidate $c$ in a voter's ranking (without changing the relative order of other candidates), it holds that
     \begin{equation*}
         f(E)\in A(c) \implies f(E')\in A(c).
     \end{equation*}
\end{definition}

Note that this axiom is incomparable to monotonicity---its guarantees are stronger for the case when $f(E)\in A(c)\setminus \{c\}$ and weaker for the case when $f(E)=c$.

If a rule satisfies \allynoharm and alliance monotonicity, then an alliance never has an incentive to discourage its members from running or performing well in elections. However, for such rules another issue becomes critical: It should be never profitable for an alliance to \emph{improve} its result by nominating several similar candidates instead of one as well. Otherwise, these rules would face the opposite problem of promoting bigger alliances, even if their size follows merely from duplicating candidates. We should require that the presence of a candidate may help his or her alliance only if it also enriches the actual offer for voters.

Formally, given an election, let us call two allies~${a_1, a_2}$ \emph{similar} to each other if for each candidate~${b}$ such that~${b}$ is an opponent of~${a_1}$ and~${a_2}$ there is no voter~${v_i}$ such that~${a_1 \succ_i b \succ_i a_2}$. Intuitively, this condition means that both candidates have the same support structure compared to candidates outside of $A(a_1)$.

Our concept of similar candidates is based on the same intuition as the idea of clones proposed by \citet{tideman1987independence}. However, both definitions are actually incomparable, as we can see in \Cref{ex:similar-vs-clone}:

\begin{example}\label{ex:similar-vs-clone}
    Consider the following election with two voters:
\begin{alignat*}{6}
    v_1 \colon \quad &\fixedwidth{b}&& \succ \fixedwidth{a_1} && \succ \fixedwidth{a_2} && \succ \fixedwidth{a_3}\\
    v_2 \colon \quad &\fixedwidth{a_2} && \succ \fixedwidth{b} && \succ \fixedwidth{a_1} && \succ \fixedwidth{a_3}
    \end{alignat*}
    Candidates~${a_1, a_2, a_3}$ are allies, while~${b}$ is their opponent. Candidates~${b}$ and~${a_1}$ are clones, yet they are not similar. On the other hand, candidates~${a_1}$ and~${a_3}$ are similar, but they are not clones.
\end{example}

Our next axiom can be viewed as an analogue of standard independence of clones, yet using the notion of similar candidates: 
\begin{definition}[\Clonedallynohelp]
    We say that a~voting rule~${f}$ is \emph{\clonedallynohelpadj} if for every elections~${E, E'}$ such that~${E'}$ is obtained from~${E}$ by removing a candidate~${c\in C}$ similar to some other candidate~${c'\in C}$, we have that:
    \begin{align*}
        \text{(1)} \quad & {f(E) \in A(c') \implies f(E') \in A(c')},\\
        \text{(2)} \quad & {f(E) \notin A(c') \implies f(E)=f(E')}.
    \end{align*}
\end{definition}

Naturally, both axioms are incomparable. However, if we assume that every set of clones corresponds to a separate alliance\footnote{Since alliances are disjoint, this assumption implies that there are no nested groups of clones in an election. We will further address this limitation in \Cref{sec:nested-alliances}.} then \clonedallynohelp implies independence of clones. 

We believe that satisfying the four aforementioned axioms is crucial in the setting with alliances. Therefore, alliance-aware voting rules satisfying them will be called \emph{basic}. Interestingly, at this point we can already exclude some certain types of alliance-aware voting rules---namely, those that extend non-monotone rules (such as STV), Copeland or scoring rules different than Plurality.

\begin{proposition}
    There exists no basic alliance-aware voting rule~${f}$ that extends (1) a non-monotone standard rule, (2) the Copeland's rule, or (3) a scoring rule different than Plurality.
\end{proposition}
\begin{proof}
The extension of each non-montone rule would violate alliance monotonicity for no-ally elections. For the Copeland's rule, it is enough to consider a no-ally election~$E$ with three candidates~$a$,~$b$,~$c$ and three voters with preferences:~${a \succ b \succ c}$,~${b \succ c \succ a}$, and~${c \succ a \succ b}$. Then, under Copeland, every candidate is a tied winner. The election is symmetric, so without loss of generality we assume that~$a$ wins. Now, consider a modified election~$E'$ in which the candidate~$c'$ is added just below~$c$ in each ranking (the election remains no-ally). Now either~$b$ or~$c$ are elected by Copeland. Consider now election~$E''$ obtained from~$E'$ by merging~$A(c)$ and~$A(c')$. From \splitting, we obtain that the winner is~$c$,~$c'$, or~$b$. However, now~$c$ and~$c'$ are similar, and by removing~$c'$ we go back to election~$E$ in which~$a$ is the winner, which violates \clonedallynohelp.

Now fix a~scoring rule~$f$, characterized by~$\vec{s} = (\ell_1, \ell_2, \ldots, \ell_m)$. Let~$\ell_2>\ell_m$ for $m>2$ (otherwise,~$f$ is Plurality). Consider now a symmetric election~$E$ with two candidates~$a, b$ and two votes:~$a \succ b$ and~$b \succ a$. Now both candidates have equal score, so assume that~$a$ wins. Now add a candidate~$b'$ to each vote, just below~$b$. Then the score of~$a$ is~$\ell_1+\ell_m$, while the score of~$b$ is~$\ell_1+\ell_2 > \ell_1+\ell_m$, hence~$b$ wins. The remaining part of the proof (merging~$A(b)$ and~$A(b')$, and removing~$b'$) is the same as in the case of Copeland.
\end{proof}

From now, we will therefore focus only on alliance-aware voting rules extending Plurality or Maximin.

\subsection{Individual-Oriented Axioms}\label{sec:individual-axioms}

Observe that the axioms considered so far only say which alliances should or should not win in certain situations, yet they say nothing about which ``precise'' candidate from the winning alliance should win. Here we have only guarantees provided by standard axioms, such as the majority or Condorcet consistency. However, one could define here also additional \emph{individual-oriented} axioms, attempting to formalize a notion of ``a candidate who is individually strong'' in the model with alliances. Consider the following election with three candidates~${\{a_1, a_2, b\}}$ and alliances~${\{\{a_1, a_2\}, \{b\}\}}$:
\begin{center}
    \begin{minipage}{0.59\linewidth}
    \begin{alignat*}{6}
    3 \text{ votes}  \colon \quad &\fixedwidth{a_1} &&\succ \fixedwidth{b} &&\succ \fixedwidth{a_2}\\
    37 \text{ votes} \colon \quad &\fixedwidth{b} &&\succ \fixedwidth{a_1} &&\succ \fixedwidth{a_2}\\
    11 \text{ votes} \colon \quad &\fixedwidth{b} &&\succ \fixedwidth{a_2} &&\succ \fixedwidth{a_1}\\
    49 \text{ votes} \colon \quad &\fixedwidth{a_2} &&\succ \fixedwidth{a_1} &&\succ \fixedwidth{b}
    \end{alignat*}

    \end{minipage}
    \begin{minipage}{0.39\linewidth}
        \begin{tikzpicture}[scale=0.9]
    \node (a1) at (0, 1) {$a_1$};
    \node (a2) at (0, -1) {$a_2$};
    \node (b) at (2, 0) {$b$};

    \draw [-{Stealth[scale=1]}] (a2) -- (a1) node [pos=.5, above, sloped] {60:40};
    \draw [-{Stealth[scale=1]}] (a1) -- (b) node [pos=.5, above, sloped] {52:48};
    \draw [-{Stealth[scale=1]}] (b) -- (a2) node [pos=.5, above, sloped] {51:49};
\end{tikzpicture}
    \end{minipage}
\end{center}
    
Suppose that we use a basic rule~${f}$ extending either Plurality or Maximin. Here, \allynoharm requires that the winner should come from the alliance~${\{a_1, a_2\}}$. However, who should be the winner in this case? We could have two contradictory intuitions here:

\begin{intuition}\label{int:key-ally}
    The winner should be~${a_1}$, since he or she won head-to-head against~${b}$, while~${a_2}$ lost. A victory against an opponent is more important to judge the quality of a candidate than a victory against an ally.
\end{intuition}
\begin{intuition}\label{int:iw}
    The winner should be~${a_2}$, since he or she would be elected by both Plurality and Maximin if there were no alliances. The additional information about alliances should be used by~${f}$ only to guarantee that the alliance~${\{a_1, a_2\}}$ wins, yet it should not affect the choice of the winner if the proper alliance wins anyway. 
\end{intuition} 

Let us first focus on \Cref{int:iw}. We formalize it as follows:

\begin{definition}[\Independentwinner consistency]\label{def:independent-winner}
Given election~${E}$ and voting rule~${f}$, we say that candidate~${c}$ is an \emph{\independentwinner} under~${f}$, if~${c=f(E')}$, where~${E'}$ is obtained from~${E}$ by splitting~${A(c)}$ into~${\{c\}}$ and~${A(c)\setminus\{c\}}$.

We say that a voting rule~${f}$ is \emph{\independentxwinner-consistent} (or, \emph{\independentwinnershort-consistent}, for short) if it elects the \independentwinner whenever one exist.
\end{definition}

Intuitively, \independentwinner is a candidate~${c}$ who would have won if he or she took part in the election as an individual candidate, not as a member of his or her alliance. 
Although in our paper we do not consider strategic behavior of candidates, we interprete this possibility as an evidence of the high voters' support for~${c}$. 

For each election and each basic alliance-aware voting rule, we can never have two \independentwinnerplural. Indeed, suppose that for some election~${E}$ and rule~${f}$ there are two \independentwinnerplural,~${a}$ and~${b}$. First, from \splitting we obtain that in the original election both~${A(a)}$ and~${A(b)}$ are winning, hence~${A(a) = A(b)}$. Second, consider the modified election~${E'}$ in which both~${a}$ and~${b}$ are excluded from~${A(a)}$, and each of them forms a singleton alliance. We can obtain this election either by excluding~$a$ first and then~$b$, or vice versa. Then, from \splitting we obtain that both~${a}$ and~${b}$ should be winners in~$E'$, a contradiction. 

On the other hand, the \independentwinner may not exist, as we can see in the following example:

\begin{example}\label{ex:no-independent-winner}
    Consider the following election $E=(C, V, \alliances)$ with~${C=\{a_1, a_2, b_1, b_2\}}$ and~${\alliances=\{\{a_1, a_2\}, \{b_1, b_2\}\}}$. Voters' preferences are the following: 
    \begin{center}
    \begin{minipage}{0.49\linewidth}
    \begin{align*}
        v_1\colon \quad b_2 \succ a_1 \succ b_1 \succ a_2\\
        v_2\colon \quad a_2 \succ b_2 \succ a_1 \succ b_1\\
        v_3\colon \quad b_1 \succ a_1 \succ a_2 \succ b_2\\
    \end{align*}
    \end{minipage}
    \begin{minipage}{0.49\linewidth}
        \begin{tikzpicture}[scale=0.8]
    \node (a1) at (0, 1) {$a_1$};
    \node (a2) at (0, -0.8) {$a_2$};
    \node (b1) at (3, 1) {$b_1$};
    \node (b2) at (3, -0.8) {$b_2$};

    \draw [-{Stealth[scale=1]}] (a1) -- (b1) node [pos=.5, above, sloped] {2:1};
    \draw [-{Stealth[scale=1]}] (b1) -- (a2) node [pos=.3, above, sloped] {2:1};
    \draw [-{Stealth[scale=1]}] (a2) -- (b2) node [pos=.5, above, sloped] {2:1};
    \draw [-{Stealth[scale=1]}] (b2) -- (a1) node [pos=.3, above, sloped] {2:1};
    \draw [-{Stealth[scale=1]}] (a1) -- (a2) node [pos=.5, above, sloped] {2:1};
    \draw [-{Stealth[scale=1]}] (b2) -- (b1) node [pos=.5, above, sloped] {2:1};
\end{tikzpicture}
    \end{minipage}
    \end{center}
    
    Consider a basic alliance-aware rule~${f}$. Let us assume that~${f}$ elects~${a_1}$ in the above election and it elects~${b_1}$ after splitting alliance~${\{a_1, a_2\}}$ into~${\{a_1\}}$ and~${\{a_2\}}$. Then there is no \independentwinner. Indeed,~${a_1}$ and~${a_2}$ are not \independentwinnerplural since they lose election after splitting the alliance~${\{a_1, a_2\}}$. On the other hand, if in the original election~${a_1}$ is the winner then, after splitting group~${\{b_1, b_2\}}$ into~${\{b_1\}}$ and~${\{b_2\}}$,~${a_1}$ is still the winner (from \splitting). Hence, also~${b_1}$ and~${b_2}$ are not \independentwinnerplural.
\end{example}

In such case, \independentwinner consistency does not specify who should win the election.

Let us now discuss \Cref{int:key-ally}. Here we propose the following definition:

\begin{definition}[\Keyally consistency]\label{def:key-ally}
Given a 2-alliance election~${E}$ and a voting rule~${f}$, we say that a candidate~${c}$ is a \emph{\keyally} under~${f}$, if~${c=f(E')}$, where~${E'}$ is obtained from~${E}$ by removing all the candidates from~${A(c)}$, except for~${c}$.

We say that a voting rule~${f}$ is \emph{\keyxally}-consistent (or, \emph{\keyallyshort-consistent}, for short) if for each election~${E}$ such that the set of \keyallyplural~${W}$ is nonempty, we have that~${f(E)\in W}$.
\end{definition}

Intuitively, a solitary winner is a candidate~$c$ who would have won if he or she were chosen as the only representative of~$A(c)$.


One could wonder why we require that an election~${E}$ is two-alliance. In \Cref{app:stronger-key-ally} we discuss the definition of this axiom without this requirement. We show that then the definition would not be satisfiable by basic alliance-aware rules extending Plurality, Maximin or the Schulze's method. At the same time, we show that such a generalization in some cases seems to be too restrictive and therefore less appealing. Intuitively, in an election with more than two alliances, a candidate might be a \keyally not because of his or her own high support, but because of the vote-splitting between the opponents.

On the contrary, \Cref{thm:ka-majority} and \Cref{thm:ka-condorcet} below illustrate that the implications of \Cref{def:key-ally} are very intuitive and natural. Note that, since we are interested only in basic alliance-aware rules (in particular, those \splittingadj), the \keyally axiom nonetheless provides certain guarantees also in elections with more than two alliances.

\begin{observation}\label{thm:ka-majority}
    If a basic alliance-aware voting rule is majority-consistent then for each election~${E}$ the set of \keyallyplural includes the subset of candidates who would be majority winners if all their allies were removed from~${E}$. 
\end{observation}
\begin{observation}\label{thm:ka-condorcet}
    If a basic alliance-aware voting rule is Condorcet-consistent then for each election~${E}$ the set of \keyallyplural is equal to the set of candidates who win head-to-head against every opponent.
\end{observation}

Directly from the above observations, we note that in some elections there might be several \keyallyplural, while in some other ones there might be none.

Since \keyally and \independentwinner consistency formalize \Cref{int:key-ally} and \Cref{int:iw} respectively, we can already note that no basic alliance-aware voting rule extending Plurality or Maximin can satisfy them both. In fact, this is true for every reasonable deterministic basic rule.

\begin{proposition}\label{thm:no-rule-iw-sw}
    No deterministic basic alliance-aware rule that is majority-consistent for elections with two candidates, is both \keyallyshort-consistent and \independentwinnershort-consistent.
\end{proposition}
\begin{proof}
    It is enough to consider a no-ally election~$E$ with three candidates~$\{a, b, c\}$ and three voters with the following preferences:~${a \succ b \succ c}$,~${b \succ c \succ a}$,~${c \succ a \succ b}$. The election is symmetric, hence we can assume that~$a$ is the winner. Consider now a modified election~$E'$ obtained from~$E$ by merging~$A(a)$ and~$A(b)$. Now, from \splitting, we know that alliance~$\{a, b\}$ wins. It means that~$a$ is the \independentwinner. However,~$a$ loses head-to-head against~$c$, while~$b$ wins, hence only~$b$ is a \keyally. 
\end{proof}

In the next section we present basic alliance-aware voting rules satisfying the \independentwinner and \keyally axioms. Apparently, all of them have simple and easy-to-explain definitions.

Which of them should be in this case preferred to be used in practice? The answer to this question depends on the context. For instance, if candidates can predict the results of the elections (for example, via polls), choose their alliances in a strategic way and bear no costs from starting individually in elections, we should definitely prefer \independentwinnershort-consistent rules. If this is not the case---for example, the membership of candidates in their alliances is fixed---we believe that \keyallyshort-consistent rules lead to more intuitive results.

\section{Alliance-Aware Voting Rules}\label{sec:rules}

In this section we present four basic alliance-aware rules. Two of them extend Plurality and two extend Maximin. Orthogonally, two of them are \keyxally-consistent, and the other two are \independentxwinner-consistent. Let us first present the key idea on which all the four rules are based.

Let~${f\in \{\text{Plurality}, \text{Maximin}\}}$ be a standard voting rule. Note that the definition of~${f}$ is based on maximizing a specifically defined~${f}$-score---the plurality score of a candidate~${c\in C}$ is equal to the number of ballots ranking~${c}$ on top, while the maximin score of~${c}$ is equal to the number of supporters in his or her most difficult head-to-head comparison (i.e., $\min_{c'\in C} \pref(c, c')$). Both definitions of score have a very interesting property.

\begin{observation}\label{obs:fscore}
    Consider an election~${E}$ and an election~${E'}$ obtained from~${E}$ by removing some candidate~${c}$. Consider now a candidate~${c'
\neq c}$. The plurality|maximin score of~${c'}$ in~${E'}$ is greater or equal to his or her score in~${E}$.
\end{observation}

Intuitively, both in case of Plurality and Maximin, the presence of additional candidates can only potentially decrease the score of~${c}$. We extend the idea of the~${f}$-score to the setting with alliances as follows:

\begin{definition}[Alliance-aware~${f}$-score]
    Let $f$ be either Plurality or Maximin. Fix an election~${E}$ and a candidate~${c}$. Let~${E'}$ be the election obtained from~${E}$ by removing all the candidates from~${A(c)\setminus\{c\}}$. The \emph{alliance-aware~${f}$-score} of~${c}$ in~${E}$ is defined as his or her standard~${f}$-score in~${E'}$. 
\end{definition}

According to the above definition, the alliance-aware plurality score is the number of votes in which no opponent of~${c}$ is ranked higher than~${c}$. On the other hand, the alliance-aware maximin score of~${c}$ is the worst-case number of supporters in his or her head-to-head comparisons against his or her opponents, i.e.,~${\min_{c'\in C\setminus A(c)} \pref(c, c')}$.

\subsection{\Independentxwinner-Consistent Rules}\label{sec:iw-rules}

Now we are ready to define the \independentwinnershort-consistent alliance-aware variant of Plurality and Maximin.

\begin{definition}[\independentwinnershort-$f$]\label{def:iw-rules}
    Let~${f\in \{\text{Plurality}, \text{Maximin}\}}$. The \independentwinnershort-$f$ rule has two rounds: The first one chooses the winning alliance and the second one chooses the winning candidate. In the first round, we compute alliance-aware~${f}$-scores for each candidate and we choose the alliance~${A}$ whose member has the highest score. In the second round, only candidates from~${A}$ can win, yet the other ones are not removed from the election. The winner is the candidate from~${A}$ with the highest standard~${f}$-score.
\end{definition}

Let us illustrate this definition via the following example:

\begin{example}\label{ex:rules-diff}
     Consider the election~$E=(C, V, \alliances)$, with six candidates~$C=\{a_1, a_2, a_3, a_4, b, c\}$ and~${\alliances=\{\{a_1, a_2, a_3, a_4\}, \{b\}, \{c\}\}}$. Voters' preferences are the following: 
    \begin{alignat*}{6}
    30 \text{ votes}\colon \quad &\fixedwidth{b}&& \succ \fixedwidth{a_2} && \succ \fixedwidth{a_4} && \succ \fixedwidth{a_3} && \succ \fixedwidth{a_1} && \succ \fixedwidth{c}\\
    35 \text{ votes}\colon \quad &\fixedwidth{a_1} && \succ \fixedwidth{a_2} && \succ \fixedwidth{a_4} && \succ \fixedwidth{a_3} && \succ \fixedwidth{b} && \succ \fixedwidth{c}\\
    5 \text{ votes}\colon \quad &\fixedwidth{a_2} && \succ \fixedwidth{a_1} && \succ \fixedwidth{a_4} && \succ \fixedwidth{a_3} && \succ \fixedwidth{b} && \succ \fixedwidth{c}\\
    30 \text{ votes}\colon \quad &\fixedwidth{a_3} && \succ \fixedwidth{c} && \succ \fixedwidth{a_4} && \succ \fixedwidth{b} && \succ \fixedwidth{a_2} && \succ \fixedwidth{a_1}
\end{alignat*}

    Now let us compute \independentwinnershort-Plurality scores. For~$b$, it is~$30$, for~$a_1$,~$a_2$ and~$a_4$ it is~$35+5=40$, for~$a_3$ it is~$35+5+30=70$, for~$c$ it is~$0$. Candidate~$a_3$ has the highest score, hence in the first round his or her alliance wins. In the second round, the winner is determined by standard Plurality, hence~$a_1$ wins with score~$35$.

    Let us now look at \independentwinnershort-Maximin. In the first round, we only care about head-to-head comparisons between opponents. The alliance-aware maximin score of~$a_1$ and~$a_2$ is~$\min(40,70)=40$ (comparisons with~$b$ and~$c$, respectively). For~$a_3$ it is~$\min(70,100)=70$, for~$a_4$ it is~$\min(70, 70)=70$. Hence, it is clear that the alliance~$A=\{a_1, a_2, a_3, a_4\}$ wins. In the second round, we additionally need to take into account the comparisons between candidates from~$A$, possibly updating their score from the first round. Candidate~$a_2$ wins against each ally: Against~$a_1$ with vote proportion~$65:35$, and against~$a_3$ and~$a_4$ with vote proportion~$70:30$. Moreover, we can check that for~$a_1, a_3$ and~$a_4$ this defeat is among the worst ones. Hence, the score of~$a_2$ remains unchanged ($40$), the score of~$a_1$ is~$35$ and the score of~$a_3$ and~$a_4$ is~$30$. Therefore,~$a_2$ is the winner.
\end{example}



Let us now study the properties of these rules. We show that they both satisfy all the four basic axioms defined in \Cref{sec:basic-axioms}, \independentwinner consistency, and the most important axiomatic properties of their standard analogues.

\begin{theorem}\label{thm:iw-f}
Let~$f\in \{\text{Plurality, Maximin}\}$. Then \independentwinnershort-$f$ is:
\begin{inparaenum}[(1)]
    \item  basic,
    \item \independentxwinner-consistent,
    \item majority-consistent,
    \item monotone.
\end{inparaenum}
Additionally, (5) \independentwinnershort-Maximin is Condorcet-consistent.
\end{theorem}
\begin{proof}
    Regarding (1): First, directly from the definition of~${f}$-score and \Cref{obs:fscore} we obtain than removing a candidate from a losing alliance does not affect the score of his or her allies at all, while the score of his or her opponents may only increase. Hence, \allynoharm is satisfied. Second, splitting a losing alliance~${A}$ into two alliances~${A_1, A_2}$ makes the~${f}$-score of every candidate~${c\in A_1}$ be counted as if he or she was the only candidate from~${A_1}$, not from~${A}$---which together with \Cref{obs:fscore} implies that his or her alliance-aware~${f}$-score decreases. On the other hand, the alliance-aware~${f}$-score of the candidates beyond~${A}$ remains unaffected. Hence, \splitting is satisfied. Third, consider two similar candidates~${a_1, a_2}$ and their opponent~${b}$. Now, the subset of votes in which~${a_1}$ is ranked higher than~${b}$ is the same as the subset of votes in which~${a_2}$ is ranked higher than~${b}$. It means that after removing~${a_2}$ from the election, both the alliance-aware plurality score and the alliance-aware maximin score of~${b}$ remain unchanged. This fact, together with \allynoharm, implies \clonedallynohelp. Fourth, alliance monotonicity follows from the fact that improving a position of a candidate~$c$ in a voter's ranking can never decrease neither the alliance-aware~$f$-score of~$c$, nor the alliance-aware~$f$-score of any ally of~$c$, nor the standard~$f$-score of~$c$. (2) follows from the fact that if a candidate~$c$ has the highest alliance-aware~$f$-score among all the candidates in the scenario where only he or she runs as an individual candidate then she would also have the highest standard~$f$-score where all the candidates run as individual ones. (3) and (5) follow from the fact that if a candidate~$c$ has standard~$f$-score~$s > \nicefrac{n}{2}$, then no opponent has greater alliance-aware~$f$-score than~$n-s < \nicefrac{n}{2}$. Then,~$A(c)$ would be elected in the first round and~$c$ would win in the second one since~$f$ is majority-consistent (Condorcet-consistent for~$f=\text{Maximin}$). For (4) the argument is the same as for alliance monotonicity.
\end{proof}

\subsection{\Keyxally-Consistent Rules}\label{sec:sw-rules}

Let us now present the extensions of Plurality and Maximin that satisfy \keyally consistency. We will again present the joint definition for both rules. 

\begin{definition}[\keyallyshort-$f$]\label{def:sw-rules}
    Let~${f\in \{\text{Plurality}, \text{Maximin}\}}$. The \keyallyshort-$f$ rule has two rounds. In the first round, we compute alliance-aware~${f}$-scores for each candidate, and we choose the set of candidates~${T}$ that obtain higher score than~${
\nicefrac{n}{2}}$. Then in the second round, we eliminate all the candidates from~${C\setminus T}$ and elect the candidate from~${T}$ with the highest standard~${f}$-score. If~$\text{ }{T=\emptyset}$, the rule terminates after the first round, returning the candidate with the highest~${f}$-score.
\end{definition}

Let us illustrate this definition via the same example as before:

\begin{example}
    Consider the same election as in \Cref{ex:rules-diff}. Let us begin with \keyallyshort-Plurality. Alliance-aware plurality scores are the same as before, yet now we directly elect~$a_3$, since he or she is the only candidate with greater score than~$\nicefrac{n}{2}$.
    
    For \keyallyshort-Maximin, the scores in the first round are the same as in \Cref{ex:rules-diff}, yet now only~$a_3$ and~$a_4$ advance to the second round, as their score is greater than~$\nicefrac{n}{2}$. In the runoff, we compare only~$a_3$ and~$a_4$ head-to-head and declare~$a_4$ the winner (with vote proportion~$70:30$).
\end{example}

Similarly as in the case of \independentwinnershort-rules, \keyallyshort-rules preserve the most important axiomatic properties of their standard analogues.

\begin{theorem}\label{thm:sw-f}
Let~$f\in \{\text{Plurality, Maximin}\}$. Then \keyallyshort-$f$ is:
\begin{inparaenum}[(1)]
    \item  basic,
    \item \keyxally-consistent,
    \item majority-consistent,
    \item monotone.
\end{inparaenum}
Additionally, (5) \keyallyshort-Maximin is Condorcet-consistent.
\end{theorem}
\begin{proof}
    For (1) the argumentation is the same as in case of \Cref{thm:iw-f}. (2) follows from the characterization of \keyallyplural presented in \Cref{thm:ka-majority} and \Cref{thm:ka-condorcet}. Additionally, for \keyallyshort-Plurality we note that if~$T=\emptyset$ and solitary winner consistency (possibly together with \splitting) requires electing~$c$, he or she needs to have the highest alliance-aware plurality score. It holds because after removing the allies of~$c$ from the election and merging the alliances of his or her opponents, the alliance-aware score of~$c$ is the same, and the opponents' score could only increase. (3) and (5) follow from the fact that the majority winner (Condorcet winner, for~$f=\text{Maximin}$) would always advance to the second round and win there. (4) follows from the fact that improving a position of a candidate~$c$ in a voter's ranking may only improve the alliance-aware~$f$-score of~$c$, does not affect the alliance-aware~$f$-score of his or her allies (hence, does not change the subset of candidates advancing to the second round), and increases the standard~$f$-score of~$c$ in the second round.
\end{proof}

To conclude this section, in \Cref{fig:overview} we present an overview of the axioms satisfied by our \independentwinnershort- and \keyallyshort- voting rules.

\begin{table}[t]
\centering
    \begin{tabular*}{\columnwidth}{ l | c | c | c | c }
    \toprule
     
     & \small{IW-Plur.} & \small{SW-Plur.} & \small{IW-Maximin} & \small{SW-Maximin} \\
     \midrule
     Ally-no-harm & \yes & \yes & \yes & \yes \\
     Res. to alliance-split. & \yes & \yes & \yes & \yes \\
     Ind. of similar allies & \yes & \yes & \yes & \yes \\
     Alliance monotonicity & \yes & \yes & \yes & \yes \\
    \midrule
     Monotonicity & \yes  & \yes  & \yes  & \yes \\
     Majority consistency & \yes  & \yes  & \yes  & \yes  \\
     Condorcet consistency  & \no & \no & \yes & \yes\\
     \midrule
     \independentwinnershort consistency   & \yes & \no & \yes & \no \\
     \keyallyshort consistency   & \no & \yes & \no & \yes \\
    \bottomrule
    \end{tabular*}
\caption{The axiomatic comparison of alliance-aware rules.}\label{fig:overview}
\end{table}

\section{Experimental Analysis}\label{sec:experiments}
In this section, we conduct two experiments, which we treat as sanity checks. In the first one we verify how often primary elections within alliances fail to select an optimal representative. In the second one, we compare alliance-aware rules with the standard ones from the utilitarian perspective. 

Since we lack the real-life election data with complete rankings and alliance affiliations, we focus on synthetically generated elections. Below we define two statistical cultures, one based on impartial culture, and the other one based on Euclidean model~\citep{EnelowHinich84,Merrill84}.
\begin{description}
    \item[Impartial Culture (IC).] Under the impartial culture model, each preference order is sampled uniformly at random. The alliances of candidates are also sampled uniformly at random.
    \item[Euclidean.] In a~${d}$-dimensional Euclidean model, each candidate, voter, and alliance is assigned an ideal point, 
    drawn independently from the uniform distribution on~${(0, 1)^{d}}$. Then, the alliance of a given candidate is the closest one (using~${\ell_2}$ distance) to to him or her. Finally, the ranking of each voter is obtained by ordering candidates according to the increasing~${\ell_2}$ distance from the voter's ideal point.
\end{description}

Note that we omit here several other well-known statistical cultures, e.g., the P\'olya-Eggenberger urn model~\citep{eggenberger1923statistik} or the Mallows model~\citep{mallowModel}. However, it seems unclear how to properly adapt them to the setting with alliances. For cultures different than impartial culture, choosing alliance membership uniformly at random does not seem appealing, since we expect that the candidates from the same alliance should be somehow ``closer'' to each other. On the other hand, for cultures different than the Euclidean one, there is no natural notion of ``distance'' between candidates. We therefore leave the study of other models for future research.

\subsection{Non-optimal Primary Winner}

\newcommand{\qq}[1]{%
    \pgfmathsetmacro\halfvalue{#1/2}
    \tikz[baseline=0.3ex]{%
        \fill[red!\halfvalue!white] (0,0ex) rectangle (6ex,2ex);
        \node[inner sep=0pt, anchor=center] at (3ex,1ex) {#1\%};%
    }%
}


 
In the following experiment we study standard Plurality and Maximin with the following modification: Only one candidate from each alliance takes part in the main election. We assume that representatives of all alliances are chosen via primaries. The primaries may be \emph{joint}, where we select primary winners using all the ballots, or \emph{disjoint}, where for each alliance~$A$ we select its representative using only the ballots of supporters of~$A$, i.e., voters who ranked a candidate from~$A$ on top.

Now we verify how frequently we witness choosing a non-optimal primary winner, i.e., a situation where there exists an alliance which is losing, but which could have won if another candidate was selected as the primary winner. 

In~\Cref{tab:spoilerity:primaries} we present the results. 
Each entry is an average over 1000 elections with 10 candidates, 101 voters and 2 alliances. (additional results with 8|10|12 candidates, 2|3 alliances, and including the Schulze's method are available in \Cref{app:experiments}). First of all, for disjoint primaries we observe a much higher probability of witnessing a non-optimal primary winner. As expected, the probability of witnessing a non-optimal primary winner is smaller under Maximin, than under Plurality---but still not negligible. The exception are 1D elections, but note that Maximin under joint primaries is Condorcet-consistent and in such elections (given the odd number of voters) the Condorcet winner always exists.

This simulation supports our claim that primaries are indeed not a good solution. 
Naturally, the exact probability of witnessing a non-optimal primary winner depends on the specific model. However, the very fact is that 
this value is significant for at least some models shows that 
this is an actual problem.


\subsection{Cost of Using Alliance-Aware Rules}
One might argue that using the alliance-aware voting rule, we can significantly decrease the social welfare. In the following experiment, we show that is not the case. 
Specifically, we compare the social welfare obtained by standard and alliance-aware rules. 

Given an election~$E=(C, V, \alliances)$, let~$c_w\in C$ be the winner of that election. We define the social welfare of candidate~$c_w$ as $\sum_{v \in V} (m - 1 - \mathrm{pos}_v(c_w))$, where~$\mathrm{pos}_v(c_w)$ denotes the position of candidate~$c_w$ in vote~$v$. In other words, the social welfare of a given candidate is equal to his or her Borda score.
The social welfare of an election~$E$ given voting rule~$f$, is defined as the social welfare of the candidate winning under rule~$f$.

In~\Cref{tab:spoilerity:primaries} we present the results. Fortunately, we see that the social welfare of alliance-aware variants of Maximin rule, provide the same social welfare as standard Maximin. For Plurality, the social welfare of alliance-aware variants is (in most cases) not only not smaller, but even larger than that of standard Plurality. Note that the Plurality with primaries is also providing larger social welfare than the standard version.

\newcommand{\ww}[1]{%
    \pgfmathsetmacro\smallvalue{(#1-450)/6}
    \tikz[baseline=0.3ex]{%
        \fill[blue!\smallvalue!white!] (0,0ex) rectangle (6ex,2ex);
        \node[inner sep=0pt, anchor=center] at (3ex,1ex) {#1};%
    }%
}



\begin{table}[t]
\centering
    \begin{tabular*}{\columnwidth}{c | c | cccc}
      \toprule
       & Rule &  IC & Euc 3D & Euc 2D & Euc 1D\\
      \midrule 
\multirow{4}{*}{\rotatebox{90}{spoilers}} & Plurality (Joint Primaries) & \qq{45.6} & \qq{10.9} & \qq{14.5} & \qq{0.3} \\ 
& Plurality (Disjoint Primaries) & \qq{49.3} & \qq{24.2} & \qq{37.5} & \qq{55.6} \\
\cmidrule(l{0pt}r{6pt}){2-6}
& Maximin (Joint Primaries) & \qq{35.3} & \qq{1.7} & \qq{1.6} & \qq{0.0} \\ 
& Maximin (Disjoint Primaries) & \qq{40.2} & \qq{11.3} & \qq{17.7} & \qq{67.1} \\
  \midrule 
\multirow{10}{*}{\rotatebox{90}{average social welfare}} & Plurality & \ww{480.8} & \ww{564.0} & \ww{523.8} & \ww{474.1} \\ 
& Plurality (Joint Primaries)  & \ww{490.3} & \ww{616.6} & \ww{595.6} & \ww{580.5} \\ 
& Plurality (Disjoint Primaries) & \ww{486.5} & \ww{588.9} & \ww{551.6} & \ww{525.8} \\ 
& IW-Plurality  & \ww{478.9} & \ww{568.1} & \ww{532.1} & \ww{493.6} \\ 
& SW-Plurality  & \ww{483.9} & \ww{626.1} & \ww{607.9} & \ww{580.8} \\ 
\cmidrule(l{0pt}r{6pt}){2-6}
& Maximin  & \ww{498.5} & \ww{643.8} & \ww{631.1} & \ww{581.6} \\ 
& Maximin (Joint Primaries) & \ww{497.0} & \ww{643.8} & \ww{631.1} & \ww{581.6} \\ 
& Maximin (Disjoint Primaries) & \ww{492.0} & \ww{628.7} & \ww{605.0} & \ww{515.0} \\ 
& IW-Maximin & \ww{497.9} & \ww{643.9} & \ww{631.0} & \ww{581.6} \\ 
& SW-Maximin & \ww{496.4} & \ww{643.8} & \ww{631.1} & \ww{581.6} \\ 
\bottomrule
\end{tabular*}
 \caption{\label{tab:spoilerity:primaries} Probability of non-optimal primary winner occurrence under Plurality and Maximin and the average social welfare of our rules compared to their classical analogues with and without primaries. Setup: 1000 elections (per culture) with 10 candidates, 101 voters, and 2 alliances.}
\end{table}

\section{Summary and Discussion}\label{sec:conclusion}

In this paper, we have proposed an idea of \emph{alliance-aware voting rules} as a proposal to eliminate the spoiler effect between candidates from the same alliance. We have introduced first axioms that are desirable in this model as well as four alliance-aware voting rules with good axiomatic properties. Our rules completely eliminate the need for alliances to run primary elections, which, as we have shown, often may return candidates that are not optimal if the goal is to win the whole election. We have argued that they provide similar social welfare to the voters as the standard rules. 

Let us complete our paper by discussing how our methods could be implemented in real-life elections and proposing one interesting extension of our model.

\subsection{Implementation Details}\label{sec:discussion}

The key factor that we find crucial about our rules is that they are easy to explain and understand. This is in contrast to many rules that are independent of clones (like Ranked Pairs proposed by \citet{tideman1987independence} or the Schulze's method). However, one could argue that for \independentwinnershort-Plurality and \keyallyshort-Plurality, there is still one issue which makes them more complicated (and therefore, less appealing) than standard Plurality: Namely, the complexity of voters' ballots. For standard Plurality it is enough for every voter to provide only his or her top choice, while our rules, as defined in \Cref{sec:rules}, seem to require full rankings.

However, it is actually not the case. For each voter~$v_i$, let us denote by~$T_i$ the largest set of candidates from the same alliance~$A$, such that~$v_i$ ranks no candidate outside of~$A$ higher than a member of~$T_i$. In words,~$T_i$ is the set of candidates who form the longest single-alliance prefix of~$v_i$'s ranking. Now we can see that for \independentwinnershort-Plurality, we in fact only need a voter~$v_i$ to provide his or her top choice together with set~$T_i$ (without even the need to order it). Indeed, the remaining information is not used by the rule neither in the first round, nor in the second one. This corresponds to a mixture of an approval and a plurality ballot.

For \keyallyshort-Plurality, the information about the~$v_i$'s preferences over candidates beyond~$T_i$ might be used in the second round. However, we can implement this method in a similar way as it is usually the case for the Plurality with runoff: The two rounds can be implemented as two separate votings. 
In the first round, each voter casts an approval ballot, approving the candidates from~$T_i$. Then candidates who gain approvals from more than~$50\%$ of voters, advance to the runoff (if there is no such candidate, the candidate with the greatest number of approvals wins). In the second round, the voters may vote using standard plurality ballots.

\subsection{Extension: Nested Alliances}\label{sec:nested-alliances}

In our work we have assumed for clarity that alliances do not overlap. However, our axioms and voting rules can be naturally extended to the setting where set~$\alliances$ forms a \emph{laminar family}, i.e., for each two alliances we have that either one of them is a subset of another, or they are disjoint. This may correspond to the situation in which we have e.g., political parties, their coalitions, and factions within them. Intuitively, now candidates $a$ and $b$ might be at the same time allies with respect to one alliance, and opponents with respect to some other (inner) one. The general idea behind our axioms remains the same, yet now their definitions need to include the fact that a candidate $c$ might be a member of multiple alliances. 

Our voting rules, instead of two rounds, now could possibly have at most~$k+1$ ones, where~$k$ is the maximal nesting depth of an alliance; intuitively, after each round we recount the alliance-aware scores of candidates, skipping the alliances that are supersets of the set of potential election winners. We present the details of our adapted definitions in \Cref{app:laminar}.

Among other possible extensions of our setting we could mention e.g., the model where alliances overlap arbitrarily or the multiwinner model. 
Besides, all the topics studied in the literature for classical single-winner voting rules (e.g., distortion \citep{procaccia2006distortion,Anshelevich2021DistortionIS}, resistance to manipulations, etc.) are interesting for alliance-aware rules as well. We leave these open problems for future research.

\section*{Acknowledgements}
This project has received funding from the European Research Council
(ERC) under the European Union’s Horizon 2020 research and innovation
programme (grant agreement No 101002854). Additionally, Grzegorz Pierczy\'nski was supported by Poland’s National Science Center grant no. 2022/45/N/ST6/00271.
    
\begin{center}
  \includegraphics[width=3cm]{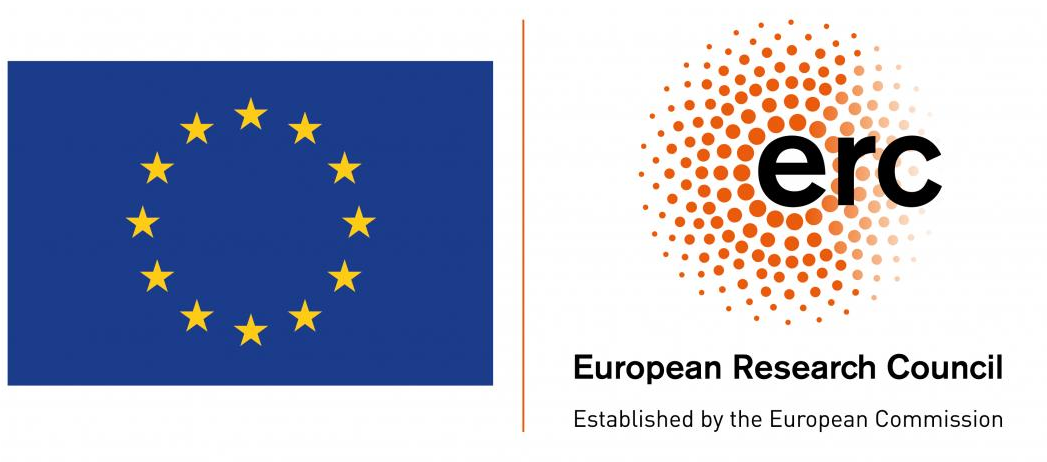}
\end{center}







\bibliographystyle{plainnat} 
\bibliography{biblio}

\newpage
\appendix 
\section{The Analysis of the Schulze's Method}\label{app:schulze}

The standard voting rule proposed by \citet{schulze2011new} can be defined as follows:
\begin{definition}[The Schulze's rule]
Let the \emph{strength} of the sequence of candidates $(c_1, c_2, \ldots, c_k)$ be defined as:
\begin{equation*}
    \min\big(\pref(c_1, c_2), \pref(c_2, c_3), \ldots, \pref(c_{k-1}, c_k)\big).
\end{equation*}
Any sequence of candidates where $a$ is the first element and $b$ is the last element would be called a \emph{path from $a$ to $b$}.

    Now, for each candidate the Schulze's rule computes the \emph{Schulze score}, defined as follows: each candidate $a$ scores one point for each candidate $b$ such that the strongest path from $a$ from $b$ is at least as strong as the strongest path from $b$ to $a$. The candidate maximizing the Schulze score is elected as the winner.
\end{definition}

If the strongest path from $a$ from $b$ is (strictly) stronger than the strongest path from $b$ to $a$, we say that $a$ \emph{dominates} $b$.

Note that if candidate $a$ dominates candidate $b$ and candidate $b$ dominates candidate $c$, then also candidate $a$ dominates candidate $c$. 
This fact implies that there always exist a candidate who scores $m-1$ points against all the other candidates.

The Schulze's rule has recently gained popularity since it has good axiomatic properties. Not only it is a Condorcet-consistent rule, but it is also one of the very few methods satisfying both monotonicity and independence of clones (the only other well-established rule with similar properties is Ranked Pairs proposed by \citet{tideman1987independence}). On the other hand, the alliance-aware version of this rule would probably not be a feasible proposal for real-life political elections---even in no-ally elections its definition is quite complex and we may expect that its alliance-aware variant would be even worse. Besides, we believe that if a rule satisfies \allynoharm and \clonedallynohelp then satisfying independence of clones is not as appealing as it is in the model without alliances. Therefore, we feel that \independentwinnershort-Maximin and \keyallyshort-Maximin are more reasonable Condorcet-consistent proposals for this model than extensions of Schulze.

Nevertheless, let us describe briefly how the idea of the Schulze's rule can be adapted to our model. Similarly as in case of Plurality and Maximin, we start by presenting the alliance-aware extension of the Schulze score.

\begin{definition}[Alliance-aware Schulze score]
    Each candidate $a$ scores one point for each his or her ally and for each opponent $b$ such that the strongest path from $a$ from $b$ containing allies of $a$ and not containing allies of $b$ is at least as strong as the strongest path from $b$ to $a$ containing allies of $b$ and not containing allies of $a$. If it is stronger, we say that $a$ dominates $b$.
\end{definition}

Intuitively, now the presence of allies can only ``help'' candidate $a$ (possibly increasing the strength of the strongest paths from him or her to some opponents) and not ``harm'' (since when calculating paths from other candidate $b$ to $a$, we cannot take into account allies of $a$).

Similarly as in case of standard Schulze score, for each election there exists a candidate who scores $m-1$ points. Otherwise, the domination relation would be intransitive: there would be three opponents, $a$, $b$ and $c$ such that $a$ dominates $b$, $b$ dominates $c$ and $c$ dominates $a$. Such a situation is impossible, which can be shown by the very same proof to the one provided by \citet{schulze2011new}: Suppose that the strength of the strongest path from $a$ to $b$ is $S_1$ and the strength of the strongest path from $b$ to $c$ is $S_2$. Then, in particular, there exists a path from $a$ to $c$ with strength $\min(S_1, S_2)$. Suppose now that there exists a path from $c$ to $a$ that has greater strength than $\min(S_1, S_2)$, say $S_3$. But if $S_1 \leq S_2$, then there is a path from $b$ to $a$ with strength $\min(S_2, S_3) \geq S_1$, a contradiction. On the other hand, if $S_1 > S_2$, then there exists a path from $c$ to $b$ with strength $\min(S_3, S_1) > S_2$, a contradiction. Note that our restriction that paths from each candidate cannot contain his or her allies, has no impact on this reasoning.

Now the definition of \independentwinnershort-Schulze is the same as \Cref{def:iw-rules}, after allowing $f\in\{\text{Plurality}, \text{Maximin}, \text{Schulze}\}$. 

\begin{definition}[\independentwinnershort-Schulze]
    The \independentwinnershort-Schulze's rule has two rounds: The first one chooses the winning alliance and the second one chooses the winning candidate. In the first round, we compute alliance-aware Schulze scores for each candidate and we choose the alliance~${A}$ whose member has the highest score. In the second round, only candidates from~${A}$ can win, yet the other ones are not removed from the election. The winner is the candidate from~${A}$ with the highest standard Schulze score.
\end{definition}

We can show that \independentwinnershort-Schulze has all the axiomatic properties of Maximin; additionally, it is cloneproof.

\begin{theorem}
    \independentwinnershort-Schulze is (1) basic, (2) monotone, (3) \independentxwinner-consistent, (4) Condorcet-consistent, (5) cloneproof.
\end{theorem}

\begin{proof}
Regarding (1): It is clear that \independentwinnershort-Schulze satisfies \splitting and \allynoharm (since the removal of a candidate $c$ does not change the strongest paths from the opponents of $c$ to any ally of $c$). \Clonedallynohelp follows from the fact that similar candidates have Besides, improving a position of a candidate $c$ can only increase his or her alliance-aware Schulze score as well as the score of the allies of $c$, while such an operation does not affect the strongest paths from the opponents of $c$ to the allies of $c$. Together with the observation that if a candidate is winning, he or she scores $m-1$ points, we obtain that if $c$ or an ally of $c$ was winning before the improvement of $c$'s position, he or she still wins after this operation---hence, both alliance monotonicity and standard monotonicity is satisfied. (3) follows from the fact that if a candidate is a winner without the help from his or her allies, he or she would also gain the highest score both in the first and the second round. (4) follows from the fact that a Condorcet winner dominates every other candidate, hence he or she would score $m-1$ points in the first round (possibly with some of his or her allies) and uniquely $m-1$ points in the second round. (5) follows from the fact that every set of clones $T$ has identical results of head-to-head comparisons with candidates outside of $T$. Hence, for each $a, b\in C$, if more than one candidate from $T$ appears on the strongest path $P$ from $a$ to $b$ then there also exists an equally strong path $P'$ from $a$ to $b$ including at most one candidate from $T$. Hence, adding or removing clones do not affect the values of the strongest paths.
\end{proof}

The definition of \keyxally-consistent version of Schulze is also similar to \Cref{def:sw-rules}, yet with one change, caused by the different nature of the Schulze score compared to maximin or plurality scores. Namely, the condition for advancing to the second round is no longer having more than $\nicefrac{n}{2}$ points: instead, we apply \Cref{thm:ka-condorcet} directly and say that the candidates who win head-to-head against every opponent advance to the second round.

\begin{definition}[\keyallyshort-Schulze]
    The \keyallyshort-Schulze's rule has two rounds. In the first round, we check if the set of candidates who win head-to-head against every opponent $T$ is nonempty. Then in the second round, we eliminate all the candidates from~${C\setminus T}$ and elect the candidate from~${T}$ with the highest standard Schulze score. If~$T=\emptyset$, the rule terminates after the first round, returning the candidate with the highest alliance-aware Schulze score.
\end{definition}

The axiomatic properties of \keyallyshort-Schulze are the same as the properties of \independentwinnershort-Schulze, yet instead of \independentwinner consistency it satisfies \keyally consistency.

\begin{theorem}
    \keyallyshort-Schulze is (1) basic, (2) monotone, (3) \keyxally-consistent, (4) Condorcet-consistent, (5) cloneproof.
\end{theorem}

\begin{proof}
The proofs of (1), (2) and (5) are the same as before. (3) follows directly from \Cref{thm:ka-condorcet}. (4) follows from the fact that the Condorcet winner would advance to the second round and score $|T|-1$ points in it via direct victories.
\end{proof}

We leave the more detailed analysis of \independentwinnershort-Schulze and \keyallyshort-Schulze for future research, yet we conjecture that other axiomatic properties of classical Schulze (e.g., reversal symmetry or Smith consistency) are also preserved.


\section{Discussion over the Definition of a \Keyally}\label{app:stronger-key-ally}

In this section, we consider an alternative definition of \keyally consistency, briefly mentioned in \Cref{sec:individual-axioms}. 

\begin{definition}[An alternative version of \keyally consistency]\label{def:key-ally2}
Given an election~${E}$ and voting rule~${f}$, we say that candidate~${c}$ is a \emph{\keyally} under~${f}$, if~${c=f(E')}$, where~${E'}$ is obtained from~${E}$ by removing from the election all the candidates from~${A(c)}$, except for~${c}$.

We say that a voting rule~${f}$ is \emph{\keyxally}-consistent (or, \emph{\keyallyshort-consistent}, for short) if for each election~${E}$ such that the set of \keyallyplural~${W}$ is nonempty, we have that~${f(E)\in W}$.
\end{definition}

\Cref{def:key-ally2} is the same as \Cref{def:key-ally}, without the restriction that an election $E$ should be two-alliance. Now, \Cref{thm:ka-majority} still holds, yet for \Cref{thm:ka-condorcet}, precisely characterizing the guarantees provided by this axiom for Condorcet-consistent rules, it is not the case. We can see this in the following example:

\begin{example}\label{ex:ka-bad-maximin}
    Consider an election $E$ with four candidates $a_1, a_2, b, c$, where only $a_1$ and $a_2$ are allies. The voters' preferences are the following:
\begin{center}
    \begin{minipage}{0.59\linewidth}
    \begin{alignat*}{6}
    5 \text{ votes}  \colon \quad &\fixedwidth{a_2} &&\succ \fixedwidth{a_1} &&\succ \fixedwidth{b} &&\succ \fixedwidth{c}\\
    5 \text{ votes} \colon \quad & \fixedwidth{c} &&\succ \fixedwidth{a_1} &&\succ \fixedwidth{b} &&\succ \fixedwidth{a_2}\\
    3 \text{ votes} \colon \quad &\fixedwidth{b} &&\succ \fixedwidth{a_2} &&\succ \fixedwidth{c}  &&\succ \fixedwidth{a_1} \\
    1 \text{ vote} \colon \quad &\fixedwidth{a_2} &&\succ \fixedwidth{b} &&\succ \fixedwidth{c}  &&\succ \fixedwidth{a_1}\\
    1 \text{ vote} \colon \quad &\fixedwidth{a_2} &&\succ \fixedwidth{b} &&\succ \fixedwidth{a_1}  &&\succ \fixedwidth{c}
    \end{alignat*}

    \end{minipage}
    \begin{minipage}{0.39\linewidth}
        \begin{tikzpicture}[scale=0.8]
    \node (a1) at (0, 2) {$a_1$};
    \node (b) at (0, -2) {$b$};
    \node (c) at (1.5, 0) {$c$};
    \node (a2) at (4, 0) {$a_2$};

    \draw [-{Stealth[scale=1]}] (a1) -- (b) node [pos=.5, above, sloped] {10:5};
    \draw [-{Stealth[scale=1]}] (b) -- (c) node [pos=.5, above, sloped] {10:5};
    \draw [-{Stealth[scale=1]}] (c) -- (a1) node [pos=.5, above, sloped] {9:6};
    \draw [-{Stealth[scale=1]}] (a2) -- (a1) node [pos=.5, above, sloped] {10:5};
    \draw [-{Stealth[scale=1]}] (a2) -- (c) node [pos=.5, above, sloped] {10:5};
    \draw [-{Stealth[scale=1]}] (b) -- (a2) node [pos=.5, above, sloped] {8:7};
\end{tikzpicture}
    \end{minipage}
\end{center}
Suppose that we use a basic rule $f$ extending Maximin. Then, after removing $a_2$ from the election, $a_1$ wins, and therefore is a \keyally. On the other hand, after removing $a_1$ from the election, $a_2$ loses ($b$ is the Condorcet winner), hence $a_2$ is not a \keyally. Hence, \Cref{def:key-ally2}, together with \allynoharm, requires that $a_1$ is elected. \Cref{def:key-ally} does not put any constraints which candidate out of $\{a_1, a_2\}$ should be elected and our \keyallyshort-Maximin rule elects $a_2$.
\end{example}

Here, we can see that the hard requirement that $a_1$ should be elected is probably not desirable. Both $a_1$ and $a_2$ win head-to-head against one of the candidates $\{b, c\}$ and lose against the other one. The only reason why $a_1$ wins after removing $a_2$ is the fortunate vote-splitting between $b$ and $c$. If the first group of $5$ voters changed their preference to $c \succ b$, the situation would be the opposite: \Cref{def:key-ally2} would require to elect $a_2$ and forbid to elect $a_1$. Hence, \Cref{ex:ka-bad-maximin} shows that after removing the two-alliance constraint, solitary winner consistency becomes less appealing.

Nevertheless, one could still be interested whether \Cref{def:key-ally2} is satisfiable for basic alliance-aware rules. While we do not know the answer to this question in general, we prove that it is not possible for any basic rule that extends Plurality, Maximin, or the Schulze's rule (see the definition in \Cref{app:schulze}). Since we have already excluded the other well-established voting rules and, as we have shown above, we are not convinced that the strengthening provided by \Cref{def:key-ally2} over \Cref{def:key-ally} is actually desirable, we treat this result as an argument to reject \Cref{def:key-ally2}.

\begin{proposition}
    There exists no basic alliance-aware voting rule that extends Plurality and satisfies \Cref{def:key-ally2}.
\end{proposition}
\begin{proof}
    Let $f$ be a basic alliance-aware voting rule that extends Plurality. Consider an election $E$ with five candidates $a_1, a_2, b, c, d$, where only $a_1$ and $a_2$ are allies. The voters' preferences are the following:
        \begin{alignat*}{6}
    10 \text{ votes}  \colon \quad &\fixedwidth{a_1}  &&\succ \fixedwidth{a_2} &&\succ \fixedwidth{b} &&\succ \fixedwidth{c} &&\succ \fixedwidth{d}\\
    35 \text{ votes}  \colon \quad &\fixedwidth{a_1} &&\succ \fixedwidth{b} &&\succ \fixedwidth{c} &&\succ \fixedwidth{d} &&\succ \fixedwidth{a_2}\\
    10 \text{ votes} \colon \quad & \fixedwidth{c} &&\succ \fixedwidth{b} &&\succ \fixedwidth{d} &&\succ \fixedwidth{a_1} &&\succ \fixedwidth{a_2}\\
    10 \text{ votes} \colon \quad &\fixedwidth{c} &&\succ \fixedwidth{d} &&\succ \fixedwidth{b}  &&\succ \fixedwidth{a_1} &&\succ \fixedwidth{a_2} \\
    35 \text{ votes} \colon \quad &\fixedwidth{a_2} &&\succ \fixedwidth{d} &&\succ \fixedwidth{c}  &&\succ \fixedwidth{b} &&\succ \fixedwidth{a_1}
    \end{alignat*}

Here, if $a_1$ and $a_2$ were not allies, $a_1$ would be selected by $f$ (since it extends Plurality). Hence, in $E$ the winner belongs to the alliance $\{a_1, a_2\}$. Now consider a modified election $E_1$ in which $b$ and $c$ are also allies. We will show that then still the winner still belongs to $\{a_1, a_2\}$. Indeed, suppose that it is not the case. Then the winner belongs to $\{b, c\}$ (otherwise, if the winner is $d$, \splitting is violated). Now suppose that $10$ voters from the fourth group improve the position of $b$ in their rankings so that $b$ is between $c$ and $d$. Then $b$ and $c$ are similar and from \clonedallynohelp we obtain that $f$ elects $c$ after removing $b$ from the election. However, then after splitting alliance $\{a_1, a_2\}$ into singletons we obtain a no-ally election in which Plurality elects $a_1$, which is contradictory to \splitting. Hence, in $E_1$ the winner belongs to $\{a_1, a_2\}$.

Now consider an election $E_1'$ obtained from $E_1$ by removing $a_2$. Then, from the analogous reasoning as for $E_1$, we have that $a_1$ is the winner. It means that $a_1$ is a solitary winner in $E_1$. On the other hand, $a_2$ is not a solitary winner in $E_1$, since after removing $a_1$ and $c$, $b$ is the winner (which together with \allynoharm implies that $b$ or $c$ is the winner in the election where we only remove $a_1$). 

Now consider an election $E_2$, obtained from $E$, in which $d$ and $c$ are allies. By the analogous reasoning as for $E_1$, we obtain that the winner belongs to $\{a_1, a_2\}$. Analogously as for $E_1$, we can show that $a_2$ is a solitary winner in $E_2$, while $a_1$ is not.

Hence, if $f$ satisfied \Cref{def:key-ally2}, it would be obliged to elect $a_1$ in $E_1$ and $a_2$ in $E_2$. However, $E_2$ can be obtained from $E_1$ by splitting alliance $\{b, c\}$ into singletons and then merging $A(c)$ and $A(d)$. Since $f$ is \splittingadj, both operations do not change the winner, a contradiction.
\end{proof}

\begin{proposition}
    There exists no basic alliance-aware voting rule that extends Maximin or the Schulze's method, and satisfies \Cref{def:key-ally2}.    
\end{proposition}
\begin{proof}
    Let $f$ be a basic alliance-aware voting rule that extends Maximin or Schulze. Consider an election $E$ with six candidates $a_1, a_2, b, c, d, e$, where only $a_1$ and $a_2$ are allies. The voters' preferences are the following:
\begin{center}
\begin{minipage}{0.45\linewidth}
    \begin{alignat*}{6}
    1 \text{ vote}  \colon \quad &\fixedwidth{b}  &&\succ \fixedwidth{a_2} &&\succ \fixedwidth{c} &&\succ \fixedwidth{e} &&\succ \fixedwidth{d} &&\succ \fixedwidth{a_1}\\
    1 \text{ vote}  \colon \quad &\fixedwidth{a_1} &&\succ \fixedwidth{e} &&\succ \fixedwidth{d} &&\succ \fixedwidth{c} &&\succ \fixedwidth{b} &&\succ \fixedwidth{a_2}\\
    1 \text{ vote}  \colon \quad &\fixedwidth{d}  &&\succ \fixedwidth{a_1} &&\succ \fixedwidth{e} &&\succ \fixedwidth{c} &&\succ \fixedwidth{b} &&\succ \fixedwidth{a_2}\\
    1 \text{ vote}  \colon \quad &\fixedwidth{a_2} &&\succ \fixedwidth{c} &&\succ \fixedwidth{b} &&\succ \fixedwidth{e} &&\succ \fixedwidth{d} &&\succ \fixedwidth{a_1}\\
    2 \text{ votes}  \colon \quad &\fixedwidth{a_2}  &&\succ \fixedwidth{e} &&\succ \fixedwidth{d} &&\succ \fixedwidth{a_1} &&\succ \fixedwidth{c} &&\succ \fixedwidth{b}\\
    2 \text{ votes}  \colon \quad &\fixedwidth{a_1} &&\succ \fixedwidth{c} &&\succ \fixedwidth{b} &&\succ \fixedwidth{a_2} &&\succ \fixedwidth{e} &&\succ \fixedwidth{d}\\
    2 \text{ votes}  \colon \quad &\fixedwidth{a_2}  &&\succ \fixedwidth{c} &&\succ \fixedwidth{b} &&\succ \fixedwidth{a_1} &&\succ \fixedwidth{e} &&\succ \fixedwidth{d}\\
    2 \text{ votes}  \colon \quad &\fixedwidth{d} &&\succ \fixedwidth{a_1} &&\succ \fixedwidth{e} &&\succ \fixedwidth{b} &&\succ \fixedwidth{a_2} &&\succ \fixedwidth{c}    
    \end{alignat*}
\end{minipage}
\begin{minipage}{0.45\linewidth}
\begin{center}
\begin{tikzpicture}
    \node (a1) at (3, 2) {$a_1$};
    \node (a2) at (3, -2) {$a_2$};
    \node (b) at (6, 0) {$b$};
    \node (c) at (4, 0) {$c$};
    \node (e) at (2, 0) {$e$};
    \node (d) at (0, 0) {$d$};

    \draw [-{Stealth[scale=1]}] (a1) -- (b) node [pos=.5, above, sloped] {8:4};
    \draw [-{Stealth[scale=1]}] (a1) -- (c) node [pos=.5, above, sloped] {8:4};
     \draw [-{Stealth[scale=1]}] (a1) -- (e) node [pos=.5, above, sloped] {8:4};
    \draw [-{Stealth[scale=1]}] (d) -- (a1) node [pos=.5, above, sloped] {7:5};

    \draw [-{Stealth[scale=1]}] (b) -- (a2) node [pos=.5, above, sloped] {7:5};
    \draw [-{Stealth[scale=1]}] (a2) -- (c) node [pos=.5, above, sloped] {8:4};
     \draw [-{Stealth[scale=1]}] (a2) -- (e) node [pos=.5, above, sloped] {8:4};
    \draw [-{Stealth[scale=1]}] (a2) -- (d) node [pos=.5, above, sloped] {8:4};

    \draw [-{Stealth[scale=1]}] (e) -- (d) node [pos=.5, above, sloped] {9:3};
    \draw [-{Stealth[scale=1]}] (c) -- (b) node [pos=.5, above, sloped] {9:3};
\end{tikzpicture}
\end{center}
\end{minipage}
\end{center}

Here, if $E$ was a no-ally election, both Maximin and Schulze would elect either $a_1$ or $a_2$. For the original $E$ it is therefore also the case. 

Consider now election $E_1$ obtained from $E$ by merging $A(c)$ and $A(b)$. In this election, the winner also belongs to $\{a_1, a_2\}$. Indeed, otherwise from \splitting we obtain that the winner belongs to $\{b, c\}$. However, then after removing $a_2$ from the election, candidates $b$ and $c$ are similar. Further, after removing $b$ we obtain a no-ally election in which both Maximin and Schulze elect $a_1$, a contradition with \allynoharm and \clonedallynohelp. Hence, in election $E_1$ the winner belongs to $\{a_1, a_2\}$. Repeating the above reasoning, we obtain that $a_1$ is a solitary winner in $E_1$, according to \Cref{def:key-ally2}. For $a_2$ it is not the case---after removing $a_1$ from the election, a candidate from $\{b, c\}$ is elected, since after further removal of $c$, $b$ becomes the Condorcet winner. Hence, \Cref{def:key-ally2} requires that $f$ should elect $a_1$ in $E_1$.

Consider now election $E_2$ obtained from $E$ by merging $A(d)$ and $A(e)$. By an analogous reasoning, we obtain that $a_2$ is a \keyally in $E_2$ according to \Cref{def:key-ally2}, while $a_1$ is not. Hence, \Cref{def:key-ally2} requires to elect $a_2$ in $E_2$. However, $E_2$ can be obtained from $E_1$ by splitting alliance $\{b, c\}$ and merging $A(d)$ and $A(e)$.  Since $f$ is \splittingadj, both operations do not change the winner, a contradiction.
\end{proof}

\section{The Model with Nested Alliances: Details}\label{app:laminar}

In this section we assume that set $\alliances$ forms a laminar family, i.e., for each $A_1, A_2\in \alliances$ ($|A_1| \leq |A_2|$) we have that either $A_1\subseteq A_2$ or $A_1\cap A_2 = \emptyset$. Let us assume that $C\notin \alliances$. 
Now we slightly change our notation: given a candidate $c\in C$, we denote by $\alliances(c)\subseteq \alliances$ the set of all the alliances such that for each $A\in \alliances(c)$ we have that $c\in A$. 

Note that now if for two candidates $a, b$ there exists an alliance $A$ such that $a, b\in A$ and there exists an alliance $A'$ such that $a\in A', b\notin A'$ then it is more difficult to say whether $a$ and $b$ are ``allies'' or ``opponents''. Our new notions are the following: we say that $a$ and $b$ are \emph{allies with respect to alliance $A$} and \emph{opponents with respect to alliance $A'$}. Intuitively, $a$ and $b$ do not want to harm each other with the confrontation against candidates outside of $A$ and at the same time, $a$ does not want to harm other members of $A'$ with the confrontation against $b$.

Now let us present the adapted definitions of our basic axioms introduced in \Cref{sec:basic-axioms} for the setting with nested alliances.

\begin{definition}[\Allynoharm]
We say that a voting rule~${f}$ satisfies \emph{\allynoharm} criterion if for every elections~${E, E'}$ such that~${E'}$ is obtained from~${E}$ by removing a candidate~${c}$ and for each $A\in \alliances(c)$, we have that
\begin{align*}
    f(E) \notin A \implies f(E') \notin A.
\end{align*}
\end{definition}

\begin{definition}[\Splitting]
    We say that a voting rule~${f}$ is \emph{\splittingadj} if for every elections~${E, E'}$ such that~${E'}$ is obtained from~${E}$ by splitting an alliance~${A}$ into two alliances~${A_1, A_2}$ (so that~${A_1\cup A_2 = A}$,~${A_1\cap A_2 = \emptyset}$ and $\alliances$ is still a laminar family) we have that
    \begin{equation*}
        f(E) \notin A \implies f(E) = f(E').
    \end{equation*} 
\end{definition}

\begin{definition}[\Clonedallynohelp]
    We say that two candidates $a, b$ are \emph{similar with respect to alliance $A$}, if $a, b\in A$ and there is no voter who ranks a candidate $c\notin A$ between $a$ and $b$.

    A voting rule~${f}$ is \emph{\clonedallynohelpadj} if for every election~$E=(C, V, \alliances)$, every alliance $A\in \alliances$, every pair of candidates~$c, c'$ similar with respect to~$A$, and an election $E'$ obtained from $E$ by removing $c$, we have that:
    \begin{align*}
        \text{(1)} \quad & {f(E) \in A \implies f(E') \in A},\\
        \text{(2)} \quad & {f(E) \notin A \implies f(E)=f(E')}.
    \end{align*}
\end{definition}

\begin{definition}[Alliance monotonicity]
     We say that a voting rule~${f}$ is \emph{alliance monotone} if for every elections~${E, E'}$ such that~${E'}$ is obtained from~${E}$ by improving the position of some candidate $c$ in a voter's ranking (without changing the relative order of other candidates) and for each alliance $A\in \alliances(c)$, it holds that
     \begin{equation*}
         f(E)\in A \implies f(E')\in A.
     \end{equation*}
\end{definition}

Intuitively, the guarantees provided by \allynoharm, \clonedallynohelp and alliance monotonicity for a candidate $c$, now hold for every alliance containing $c$, not just one.

The notion of \independentwinner consistency remains the same as in the setting with disjoint alliances: The \independentwinner is the candidate who would have won the election if he or she have run as an individual candidate.

\begin{definition}[\Independentwinner consistency]
Given election~${E}$ and voting rule~${f}$, we say that candidate~${c}$ is an \emph{\independentwinner} under~${f}$, if~${c=f(E')}$, where~${E'}$ is obtained from~${E}$ by removing $c$ from every alliance $A$ such that $c\in A$ and adding a singleton alliance $\{c\}$.

We say that a voting rule~${f}$ is \emph{\independentxwinner-consistent} (or, \emph{\independentwinnershort-consistent}, for short) if it elects the \independentwinner whenever one exist.
\end{definition}

The extension of the definition of \keyally consistency is a bit more involved. Note that in the model with nested alliances, elections with merely two alliances would be significantly more specific than in the model with disjoint ones. However, note that for rules satisfying \splitting, the notion of \keyally in \Cref{def:key-ally} can be equivalently rephrased as follows: Given an election~${E}$ (possibly not two-alliance) and a voting rule~${f}$, we say that a candidate~${c}$ is a \emph{\keyally} under~${f}$, if~${c=f(E')}$, where~${E'}$ is obtained from~${E}$ by removing from the election all the candidates from~${A}$, except for~${c}$, and merging all the alliances of the opponents into one alliance $C\setminus A$. Such a definition would be equivalent to \Cref{def:key-ally} for all the basic alliance-aware rules. We define \keyally consistency for the setting with nested alliances in the similar spirit: 

\begin{definition}[\Keyally with respect to alliance $A$]
Given an election~${E}$ and a voting rule~${f}$, we say that a candidate~${c}$ is a \emph{\keyally} under~${f}$ with respect to alliance $A\in \alliances(c)$, if~${c=f(E')}$, where~${E'}$ is obtained from~${E}$ by removing from the election all the candidates from~${A}$, except for~${c}$, and adding an alliance $C\setminus A$. 
\end{definition}

Note that, if for an election $E$, there are two candidates $a, b$ such that $a$ is a \keyally with respect to alliance $A_1$ and $b$ is a \keyally with respect to alliance $A_2$, then either $A_1\subseteq A_2$ or vice versa. Now \keyally consistency requires that among the existing \keyallyplural, the one who is a \keyally with respect to the smallest alliance $A_{\mathrm{min}}$ is elected.

Note that if we assume that for each election $E=(C, V, \alliances)$ and candidate $c\in C$ it holds that $\{c\}\in \alliances$, our axioms become stronger than their disjoint-alliance counterparts. Specifically, alliance monotonicity implies montonicity and \keyally consistency implies majority consistency (Condorcet consistency) for each basic alliance-aware rule extending a majority-consistent (Condorcet-consistent) rule.

From now, let us then keep the aforementioned assumption (if this is not the case, modify the considered election $E$ accordingly) to simplify the description of our adapted alliance-aware voting rules. 

We start with presenting the extension of the idea of alliance-aware~$f$-score:

\begin{definition}[Alliance-aware~${f}$-score]
    Let~${f\in \{\text{Plurality}, \text{Maximin}\}}$. Fix an election~${E}$ and a candidate~${c}$. Let~${E'}$ be the election obtained from~${E}$ by removing all the candidates from~${A\setminus\{c\}}$ where $A$ is the greatest alliance containing $c$. The \emph{alliance-aware~${f}$-score} of~${c}$ in~${E}$ is defined as his or her standard~${f}$-score in~${E'}$. 
\end{definition}

Now we are ready to present the definitions of the \independentwinnershort- and \keyallyshort- rules:

\begin{definition}[\independentwinnershort-$f$]
    Let~${f\in \{\text{Plurality}, \text{Maximin}\}}$. Initially, the set of potential winner $W$ equals $C$. The \independentwinnershort-$f$ rule proceeds in rounds: After each round, the size of $W$ decreases. In each round, we compute alliance-aware~${f}$-scores for each candidate and we choose the greatest alliance~${A}$ of a candidate from $W$ who has the highest score. Then we set~$W=A$, delete alliance $A$ (without removing any candidates), and repeat until set $W$ becomes a singleton. 
\end{definition}

\begin{definition}[\keyallyshort-$f$]
    Let~${f\in \{\text{Plurality}, \text{Maximin}\}}$. The \keyallyshort-$f$ rule proceeds in rounds: After each round, the size of $W$ decreases. In each round, we compute alliance-aware~${f}$-scores for each candidate, and select the set of candidates $T$ who scored more than $\nicefrac{n}{2}$ points. If there is no such candidate, the rule terminates and returns the candidate with the highest alliance-aware~${f}$-score. Otherwise, all the candidates from~$C\setminus T$ are removed, together with all the alliances $A\in \alliances$ such that $T \subseteq A$, and the procedure is repeated.
\end{definition}

It is straightforward to check that our rules have the same axiomatic properties as their disjoint-alliance counterparts (the argumentation is analogous as in the proofs of \Cref{thm:iw-f} and \Cref{thm:sw-f}, yet for higher number of rounds; note that a pair of candidates $a, b$ who are allies with respect to alliance $A$ may start to ``harm'' each other only when no candidate outside of $A$ has a chance of winning).


\section{Experiments: Full Results}\label{app:experiments}

In~\Cref{tab:spoilerity:primaries_all} we present additional results for probability of non-optimal primary winner occurrence. We focus on Plurality, Maximin, and Schulze (see the definition in \Cref{app:schulze}). Each entry is an average over 1000 elections with 8|10|12 candidates, 101 voters, and 2|3 alliances.

In~\Cref{tab:social_welfare_all_1} and~\Cref{tab:social_welfare_all_2} we present additional analysis of social welfare of Plurality, Maximin, and Schulze. Each entry is an average over 1000 elections with 8|10|12 candidates, 101 voters, and 2 and 3 alliances, respectively.

\begin{table}[t]
\small
\centering
\begin{tabular}{c | cc | cccc}
  \toprule
   & |A| & |C| & IC & Euc 3D & Euc 2D & Euc 1D\\
     \midrule 
Plurality (Joint Primaries) & 2 & 8 & 33.9\% & 7.2\% & 9.0\% & 0.6\% \\ 
Plurality (Disjoint Primaries) & 2 & 8 & 37.9\% & 16.1\% & 28.3\% & 31.4\% \\ 
Maximin (Joint Primaries) & 2 & 8 & 28.2\% & 1.4\% & 0.8\% & 0.0\% \\ 
Maximin (Disjoint Primaries) & 2 & 8 & 30.9\% & 10.0\% & 14.7\% & 31.7\% \\ 
Schulze (Joint Primaries) & 2 & 8 & 28.4\% & 1.4\% & 0.8\% & 0.0\% \\ 
Schulze (Disjoint Primaries) & 2 & 8 & 30.2\% & 9.8\% & 14.6\% & 31.7\% \\ 
\midrule
Plurality (Joint Primaries) & 2 & 10 & 45.6\% & 10.9\% & 14.5\% & 0.0\% \\ 
Plurality (Disjoint Primaries) & 2 & 10 & 49.3\% & 24.2\% & 37.5\% & 55.6\% \\ 
Maximin (Joint Primaries) & 2 & 10 & 35.3\% & 1.7\% & 1.6\% & 0.0\% \\ 
Maximin (Disjoint Primaries) & 2 & 10 & 40.2\% & 11.3\% & 17.7\% & 67.1\% \\ 
Schulze (Joint Primaries) & 2 & 10 & 34.8\% & 1.6\% & 1.6\% & 0.0\% \\ 
Schulze (Disjoint Primaries) & 2 & 10 & 41.0\% & 11.2\% & 17.7\% & 66.9\% \\ 
\midrule
Plurality (Joint Primaries) & 2 & 12 & 52.3\% & 15.0\% & 17.0\% & 0.2\% \\ 
Plurality (Disjoint Primaries) & 2 & 12 & 55.9\% & 31.3\% & 44.6\% & 35.4\% \\ 
Maximin (Joint Primaries) & 2 & 12 & 39.6\% & 1.8\% & 1.9\% & 0.0\% \\ 
Maximin (Disjoint Primaries) & 2 & 12 & 46.2\% & 13.6\% & 24.6\% & 36.6\% \\ 
Schulze (Joint Primaries) & 2 & 12 & 38.8\% & 1.8\% & 1.9\% & 0.0\% \\ 
Schulze (Disjoint Primaries) & 2 & 12 & 46.1\% & 13.9\% & 24.7\% & 36.5\% \\ 
\midrule
Plurality (Joint Primaries) & 3 & 8 & 31.9\% & 9.0\% & 9.6\% & 1.8\% \\ 
Plurality (Disjoint Primaries) & 3 & 8 & 26.3\% & 18.2\% & 25.6\% & 21.5\% \\ 
Maximin (Joint Primaries) & 3 & 8 & 23.0\% & 1.7\% & 0.9\% & 0.0\% \\ 
Maximin (Disjoint Primaries) & 3 & 8 & 28.1\% & 11.7\% & 23.8\% & 23.4\% \\ 
Schulze (Joint Primaries) & 3 & 8 & 22.0\% & 1.7\% & 0.9\% & 0.0\% \\ 
Schulze (Disjoint Primaries) & 3 & 8 & 27.4\% & 11.7\% & 23.8\% & 23.4\% \\ 
\midrule
Plurality (Joint Primaries) & 3 & 10 & 39.8\% & 14.0\% & 13.4\% & 1.7\% \\ 
Plurality (Disjoint Primaries) & 3 & 10 & 38.8\% & 27.0\% & 33.9\% & 26.4\% \\ 
Maximin (Joint Primaries) & 3 & 10 & 31.3\% & 1.7\% & 0.8\% & 0.0\% \\ 
Maximin (Disjoint Primaries) & 3 & 10 & 36.8\% & 17.8\% & 29.4\% & 30.2\% \\ 
Schulze (Joint Primaries) & 3 & 10 & 30.0\% & 1.8\% & 0.8\% & 0.0\% \\ 
Schulze (Disjoint Primaries) & 3 & 10 & 36.1\% & 17.4\% & 29.4\% & 30.1\% \\ 
\midrule
Plurality (Joint Primaries) & 3 & 12 & 40.5\% & 15.0\% & 21.0\% & 0.5\% \\ 
Plurality (Disjoint Primaries) & 3 & 12 & 43.5\% & 33.3\% & 42.4\% & 30.1\% \\ 
Maximin (Joint Primaries) & 3 & 12 & 32.8\% & 1.8\% & 2.2\% & 0.0\% \\ 
Maximin (Disjoint Primaries) & 3 & 12 & 42.3\% & 20.2\% & 32.2\% & 35.2\% \\ 
Schulze (Joint Primaries) & 3 & 12 & 32.0\% & 1.8\% & 2.2\% & 0.0\% \\ 
Schulze (Disjoint Primaries) & 3 & 12 & 41.6\% & 19.9\% & 31.9\% & 35.0\% \\ 

  \bottomrule
\end{tabular}
 \caption{\label{tab:spoilerity:primaries_all} Probability of non-optimal primary winner occurrence. Setup: 1000 elections (per culture)  with  8|10|12 candidates, 101 voters, and 2|3 alliances.}
\end{table}


\begin{table}[t]
\small
\centering
\begin{tabular}{c | cc | cccc}
  \toprule
   & |A| & |C| & IC & Euc 3D & Euc 2D & Euc 1D\\
   \midrule
Plurality & 2 & 8 & 375.4 & 445.1 & 413.6 & 378.5 \\ 
Plurality (Joint Primaries) & 2 & 8 & 381.7 & 477.0 & 465.7 & 438.9 \\ 
Plurality (Disjoint Primaries) & 2 & 8 & 379.1 & 461.8 & 437.4 & 408.1 \\ 
IW-Plurality & 2 & 8 & 372.7 & 448.0 & 420.2 & 393.0 \\ 
SW-Plurality & 2 & 8 & 376.4 & 479.8 & 471.7 & 441.1 \\ 
Maximin & 2 & 8 & 386.2 & 494.3 & 487.3 & 456.0 \\ 
Maximin (Joint Primaries) & 2 & 8 & 384.7 & 494.1 & 487.2 & 456.0 \\ 
Maximin (Disjoint Primaries) & 2 & 8 & 382.2 & 484.3 & 469.4 & 419.2 \\ 
IW-Maximin & 2 & 8 & 385.9 & 494.3 & 487.3 & 456.0 \\ 
SW-Maximin & 2 & 8 & 385.1 & 494.1 & 487.3 & 456.0 \\ 
Schulze & 2 & 8 & 386.5 & 494.4 & 487.3 & 456.0 \\ 
Schulze (Joint Primaries) & 2 & 8 & 384.8 & 494.1 & 487.2 & 456.0 \\ 
Schulze (Disjoint Primaries) & 2 & 8 & 382.5 & 484.3 & 469.5 & 419.2 \\ 
    \midrule 
Plurality & 2 & 10 & 480.8 & 564.0 & 523.8 & 474.1 \\ 
Plurality (Joint Primaries) & 2 & 10 & 490.3 & 616.6 & 595.6 & 580.5 \\ 
Plurality (Disjoint Primaries) & 2 & 10 & 486.5 & 588.9 & 551.6 & 525.8 \\ 
IW-Plurality & 2 & 10 & 478.9 & 568.1 & 532.1 & 493.6 \\ 
SW-Plurality & 2 & 10 & 483.9 & 626.1 & 607.9 & 580.8 \\ 
Maximin & 2 & 10 & 498.5 & 643.8 & 631.1 & 581.6 \\ 
Maximin (Joint Primaries) & 2 & 10 & 497.0 & 643.8 & 631.1 & 581.6 \\ 
Maximin (Disjoint Primaries) & 2 & 10 & 492.0 & 628.7 & 605.0 & 515.0 \\ 
IW-Maximin & 2 & 10 & 497.9 & 643.9 & 631.0 & 581.6 \\ 
SW-Maximin & 2 & 10 & 496.4 & 643.8 & 631.1 & 581.6 \\ 
Schulze & 2 & 10 & 498.8 & 643.9 & 631.1 & 581.6 \\ 
Schulze (Joint Primaries) & 2 & 10 & 497.1 & 643.9 & 631.1 & 581.6 \\ 
Schulze (Disjoint Primaries) & 2 & 10 & 491.9 & 628.9 & 605.0 & 514.9 \\ 
    \midrule 
Plurality & 2 & 12 & 586.2 & 680.6 & 629.5 & 585.2 \\ 
Plurality (Joint Primaries) & 2 & 12 & 600.5 & 751.2 & 730.4 & 674.7 \\ 
Plurality (Disjoint Primaries) & 2 & 12 & 595.2 & 713.7 & 668.3 & 633.6 \\ 
IW-Plurality & 2 & 12 & 584.2 & 688.7 & 642.1 & 606.4 \\ 
SW-Plurality & 2 & 12 & 592.1 & 765.4 & 742.4 & 675.6 \\ 
Maximin & 2 & 12 & 610.2 & 789.7 & 773.4 & 711.5 \\ 
Maximin (Joint Primaries) & 2 & 12 & 609.1 & 789.4 & 773.2 & 711.5 \\ 
Maximin (Disjoint Primaries) & 2 & 12 & 602.4 & 768.7 & 733.7 & 655.1 \\ 
IW-Maximin & 2 & 12 & 609.7 & 789.6 & 773.4 & 711.5 \\ 
SW-Maximin & 2 & 12 & 608.5 & 789.5 & 773.2 & 711.5 \\ 
Schulze & 2 & 12 & 611.0 & 789.7 & 773.5 & 711.5 \\ 
Schulze (Joint Primaries) & 2 & 12 & 609.2 & 789.3 & 773.3 & 711.5 \\ 
Schulze (Disjoint Primaries) & 2 & 12 & 602.6 & 768.9 & 733.8 & 654.9 \\
    
  \bottomrule
\end{tabular}
 \caption{\label{tab:social_welfare_all_1} Average social welfare. Setup: 1000 elections (per culture)  with  8|10|12 candidates, 101 voters, and 2 alliances.}
\end{table}

\begin{table}[t]
\small
\centering
\begin{tabular}{c | cc | cccc}
  \toprule
   & |A| & |C| & IC & Euc 3D & Euc 2D & Euc 1D\\
   \midrule
Plurality & 3 & 8 & 374.8 & 443.4 & 410.7 & 376.2 \\ 
Plurality (Joint Primaries) & 3 & 8 & 381.8 & 476.7 & 463.0 & 429.5 \\ 
Plurality (Disjoint Primaries) & 3 & 8 & 379.6 & 462.6 & 435.1 & 398.1 \\ 
IW-Plurality & 3 & 8 & 373.4 & 448.5 & 420.0 & 378.6 \\ 
SW-Plurality & 3 & 8 & 376.0 & 474.7 & 463.5 & 433.5 \\ 
Maximin & 3 & 8 & 386.5 & 492.9 & 486.0 & 455.1 \\ 
Maximin (Joint Primaries) & 3 & 8 & 385.7 & 492.9 & 486.0 & 455.1 \\ 
Maximin (Disjoint Primaries) & 3 & 8 & 382.7 & 480.7 & 460.2 & 440.8 \\ 
IW-Maximin & 3 & 8 & 386.3 & 492.9 & 486.1 & 455.1 \\ 
SW-Maximin & 3 & 8 & 385.4 & 492.9 & 486.0 & 455.1 \\ 
Schulze & 3 & 8 & 386.9 & 492.9 & 486.1 & 455.1 \\ 
Schulze (Joint Primaries) & 3 & 8 & 385.9 & 492.9 & 486.0 & 455.1 \\ 
Schulze (Disjoint Primaries) & 3 & 8 & 383.0 & 480.9 & 460.1 & 440.9 \\ 
\midrule
Plurality & 3 & 10 & 479.9 & 562.6 & 527.0 & 480.3 \\ 
Plurality (Joint Primaries) & 3 & 10 & 490.3 & 619.5 & 598.9 & 552.6 \\ 
Plurality (Disjoint Primaries) & 3 & 10 & 486.4 & 592.8 & 560.6 & 515.0 \\ 
IW-Plurality & 3 & 10 & 478.0 & 567.5 & 535.4 & 482.0 \\ 
SW-Plurality & 3 & 10 & 481.3 & 611.9 & 596.2 & 558.2 \\ 
Maximin & 3 & 10 & 497.1 & 642.7 & 631.1 & 583.5 \\ 
Maximin (Joint Primaries) & 3 & 10 & 495.7 & 642.6 & 631.0 & 583.5 \\ 
Maximin (Disjoint Primaries) & 3 & 10 & 491.9 & 621.4 & 591.4 & 562.6 \\ 
IW-Maximin & 3 & 10 & 496.8 & 642.8 & 631.2 & 583.5 \\ 
SW-Maximin & 3 & 10 & 495.7 & 642.9 & 631.1 & 583.5 \\ 
Schulze & 3 & 10 & 497.6 & 643.0 & 631.2 & 583.5 \\ 
Schulze (Joint Primaries) & 3 & 10 & 496.0 & 642.6 & 631.0 & 583.5 \\ 
Schulze (Disjoint Primaries) & 3 & 10 & 492.2 & 622.2 & 591.3 & 562.7 \\ 
\midrule
Plurality & 3 & 12 & 585.7 & 677.0 & 630.6 & 571.5 \\ 
Plurality (Joint Primaries) & 3 & 12 & 601.6 & 760.2 & 728.6 & 675.7 \\ 
Plurality (Disjoint Primaries) & 3 & 12 & 595.1 & 716.9 & 671.1 & 621.3 \\ 
IW-Plurality & 3 & 12 & 584.9 & 685.3 & 639.4 & 570.0 \\ 
SW-Plurality & 3 & 12 & 590.1 & 748.6 & 730.9 & 680.4 \\ 
Maximin & 3 & 12 & 610.2 & 789.6 & 772.6 & 708.7 \\ 
Maximin (Joint Primaries) & 3 & 12 & 608.8 & 789.1 & 772.6 & 708.7 \\ 
Maximin (Disjoint Primaries) & 3 & 12 & 602.3 & 762.1 & 722.9 & 680.8 \\ 
IW-Maximin & 3 & 12 & 610.1 & 789.4 & 772.5 & 708.7 \\ 
SW-Maximin & 3 & 12 & 608.7 & 789.1 & 772.5 & 708.7 \\ 
Schulze & 3 & 12 & 611.0 & 789.4 & 772.6 & 708.7 \\ 
Schulze (Joint Primaries) & 3 & 12 & 609.0 & 789.1 & 772.6 & 708.7 \\ 
Schulze (Disjoint Primaries) & 3 & 12 & 602.7 & 762.2 & 723.2 & 680.6 \\ 

  \bottomrule
\end{tabular}
 \caption{\label{tab:social_welfare_all_2} Average social welfare. Setup: 1000 elections (per culture) with  8|10|12 candidates, 101 voters, and 3 alliances.}
\end{table}

\end{document}